\documentclass{article}

\usepackage{ijcai18}
\usepackage{times}
\usepackage{xcolor}
\usepackage{soul}
\usepackage{amsmath}
\usepackage{amssymb}
\usepackage{ifthen}
\usepackage{cleveref}
\usepackage{tikz}
\usepackage{stmaryrd}
\usetikzlibrary{topaths,calc}
\usetikzlibrary{arrows}
\usepackage{amsthm}
\usepackage[inline]{enumitem}
\usepackage{thmtools}
\usepackage[utf8]{inputenc}
\usepackage[small]{caption}
\usepackage[disable]{todonotes}

\newcommand{\consts}{\ensuremath{\mathbf{C}}}
\newcommand{\nulls}{\ensuremath{\mathbf{N}}}
\newcommand{\vars}{\ensuremath{\mathbf{V}}}
\newcommand{\inst}[1]{\ensuremath{\fk{#1}}}
\newcommand{\class}[1]{\ensuremath{\mathsf{#1}}}
\newcommand{\ont}[1]{\ensuremath{\mathcal{#1}}}
\newcommand{\sig}[1]{\ensuremath{\mathsf{sig}(#1)}}
\newcommand{\var}[1]{\ensuremath{\mathsf{var}(#1)}}
\newcommand{\forew}[2]{\ensuremath{\mathsf{FORew(#1,#2)}}}
\newcommand{\width}[1]{\ensuremath{\mathsf{wd}(#1)}}
\newcommand{\threeexp}{\textsc{3ExpTime}}
\newcommand{\dec}[1]{\ensuremath{\llbracket #1 \rrbracket}}
\newcommand{\height}[1]{\ensuremath{\mathrm{hgt}(#1)}}
\newcommand{\tw}[1]{\ensuremath{\mathsf{tw}(#1)}}
\newcommand{\dist}{\ensuremath{\mathrm{dist}}}
\newcommand{\parity}{\ensuremath{\mathrm{parity}}}
\newcommand{\names}[1]{\ensuremath{\mathrm{names}(#1)}}
\newcommand{\dir}{\ensuremath{\mathrm{Dir}}}
\newcommand{\ndia}[1]{\ensuremath{\langle #1 \rangle}}
\newcommand{\nbox}[1]{\ensuremath{[#1]}}
\newcommand{\translation}[2]{\ensuremath{\eta_{#2}(#1)}}
\newcommand{\limpl}{\ensuremath{\rightarrow}}
\newcommand{\adom}[1]{\ensuremath{\mathsf{adom}(#1)}}
\newcommand{\dom}[1]{\ensuremath{\mathsf{dom}(#1)}}
\newcommand{\free}[1]{\ensuremath{\mathrm{free}(#1)}}

\newcommand{\val}[1]{\ensuremath{\llbracket #1 \rrbracket}}

\newcommand{\sche}[1]{\ensuremath{\mathbf{#1}}}

\newcommand{\tup}[1]{\ensuremath{{(#1)}}}
\newcommand{\set}[1]{\ensuremath{\{#1\}}}
\newcommand{\db}[1]{\ensuremath{\fk{#1}}}
\newcommand{\ca}[1]{\ensuremath{\mathcal{#1}}}

\newcommand{\fk}[1]{\ensuremath{\mathfrak{#1}}}

\newcommand{\ve}[1]{\ensuremath{\bar{#1}}}

\renewcommand{\land}{\ensuremath{\wedge}}

\newcommand{\mbb}[1]{\ensuremath{\mathbb{#1}}}

\newcommand{\discnode}{\ensuremath{\sharp}}
\newcommand{\twoexp}{\textsc{2ExpTime}}
\newcommand{\expo}{\textsc{ExpTime}}
\newcommand{\size}[1]{\ensuremath{\lVert #1 \rVert}}

\newcommand{\cost}[2]{\ensuremath{\mathrm{cost}(#1,#2)}}
\newcommand{\costq}[1]{\ensuremath{\mathrm{cost}(#1)}}
\newcommand{\act}{\ensuremath{\mathrm{Act}}}
\newcommand{\obj}{\ensuremath{\mathrm{Obj}}}

\newcommand{\goal}{\ensuremath{\mathrm{goal}}}

\newcommand{\ata}[1]{\ensuremath{\ca{#1}}}
\newcommand{\ac}[1]{\ensuremath{\mathtt{#1}}}

\def\markfull{\vrule height 4pt width 4pt depth 0pt}

\title{First-Order Rewritability of Frontier-Guarded Ontology-Mediated Queries}

\author{
Pablo Barcel\'o$^1$,
Gerald Berger$^2$,
Carsten Lutz$^3$ and
Andreas Pieris$^4$\\
$^1$ Millennium Institute for Foundational Research on Data \& DCC, University of Chile\\
$^2$ Institute of Logic and Computation, TU Wien\\
$^3$ Department of Mathematics and Computer Science, University of Bremen\\
$^4$ School of Informatics, University of Edinburgh
}

\declaretheoremstyle[
  headfont=\bfseries,
  notefont=\bfseries, notebraces={\bfseries (}{\bfseries )},
  bodyfont=\itshape,
  postheadspace=.5em,
  spaceabove=6pt,
  spacebelow=6pt,
  headpunct={.}
]{thstyle}

\declaretheoremstyle[
  headfont=\bfseries,
  notefont=\normalfont, notebraces={\bfseries (}{\bfseries )},
  bodyfont=\normalfont,
  spaceabove=6pt,
  spacebelow=6pt,
  postheadspace=.5em,
  headpunct={.}
]{defstyle}

\declaretheoremstyle[
  headfont=\itshape,
  notefont=\normalfont, notebraces={(}{)},
  bodyfont=\normalfont,
  spaceabove=6pt,
  spacebelow=9pt,
  postheadspace=.5em,
  headpunct={.}
]{remarkstyle}

\makeatletter
\declaretheoremstyle[
   preheadhook=\renewcommand\@upn{},
   name=Example,
   spaceabove=6pt, spacebelow=6pt,
   headindent=0pt, postheadspace=1ex,
   headfont=\itshape,notefont=\itshape, headpunct=.,
   bodyfont=\normalfont
]{example}
\declaretheorem[style=example,refname={example,examples},Refname={Example,Examples}]{example}
\makeatother

\declaretheoremstyle[
  headfont=\itshape,
  notefont=\itshape, notebraces={(}{)},
  bodyfont=\normalfont,
  spaceabove=6pt,
  spacebelow=9pt,
  postheadspace=\labelsep,
  headpunct={.}
]{proofstyle}

\declaretheorem[style=proofstyle,unnumbered,qed=$\square$]{proof}
\declaretheorem[style=remarkstyle,numbered=no]{remark}
\declaretheorem[style=thstyle,refname={theorem,theorems},Refname={Theorem,Theorems}]{theorem}

\declaretheorem[name=Lemma,style=thstyle,sibling=theorem,refname={lemma,lemmas},Refname={Lemma,Lemmas}]{lemma}
\declaretheorem[style=thstyle,sibling=theorem,refname={corollary,corollaries},Refname={Corollary,Corollaries}]{corollary}

\declaretheorem[style=thstyle,numbered=no,refname={fact,facts},Refname={Claim,Claims}]{claim}

\newif\iflong
\longtrue % include the appendix

\begin{document}
\maketitle

\begin{abstract}
We focus on ontology-mediated queries (OMQs) based on (frontier-)guarded existential rules and (unions of) conjunctive queries, and we investigate the problem of FO-rewritability, i.e., whether an OMQ can be rewritten as a first-order query. We adopt two different approaches.
The first approach employs standard two-way alternating parity tree automata. Although it does not lead to a tight complexity bound, it provides a transparent solution based on widely known tools.
The second approach relies on a sophisticated automata model, known as cost automata. This allows us to show that our problem is \twoexp-complete.
In both approaches, we provide semantic characterizations of FO-rewritability that are of independent interest.
\end{abstract}

\section{Introduction}\label{sec:introduction}

{\em Ontology-based data access} (OBDA) is a successful application of KRR technologies in information management systems~\cite{PLCG+08}. One premier goal is to facilitate access to data that is heterogeneous and incomplete. This is achieved via an ontology that enriches the user query, typically a union of conjunctive queries, with domain knowledge.
It turned out that the ontology and the user query can be seen as two components of one composite query, called {\em ontology-mediated query} (OMQ)~\cite{BCLW14}. The problem of answering OMQs is thus central to OBDA.

Building ontology-aware database systems from scratch, with sophisticated optimization techniques, is a non-trivial task that requires a great effort. An important route towards practical implementation of OMQ answering is thus to use conventional database management systems. The problem that such systems are unaware of ontologies can be addressed by query rewriting: the ontology $\ont{O}$ and the database query $q$ are combined into a new query $q_{\ont{O}}$, the so-called rewriting, which gives the same answer as the OMQ consisting of $\ont{O}$ and $q$ over all input databases. It is of course essential that the rewriting $q_{\ont{O}}$ is expressed in a language that can be handled by standard database systems. The typical language that is considered in this setting is first-order (FO) queries.

Although in the OMQ setting description logics (DLs) are often used for modeling ontologies, it is widely accepted that for handling arbitrary arity relations in relational databases it is convenient to use {\em tuple-generating dependencies} (TGDs), a.k.a.~{\em existential rules} or {\em Datalog$^{\pm}$ rules}.
It is known, however, that evaluation of rule-based OMQs is undecidable~\cite{CaGK13}. This has led to a flurry of activity for identifying restrictions on TGDs that lead to decidability.
The main decidable classes are (i) {\em (frontier-)guarded} TGDs~\cite{BLMS11,CaGK13}, which includes {\em linear} TGDs~\cite{CaGL12}, (ii) {\em acyclic} sets of TGDs~\cite{FKMP05}, and (iii) {\em sticky} sets of TGDs~\cite{CaGP12}.
There are also extensions that capture Datalog; see
the same references.

For OMQs based on linearity, acyclicity, and stickiness, FO-rewritings are always guaranteed to exist~\cite{GoOP14}. In contrast, there are (frontier-)guarded OMQs that are inherently recursive, and thus not expressible as a first-order query. This brings us to our main question: {\em Can we check whether a (frontier-)guarded OMQ is FO-rewritable?} Notice that for OMQs based on more expressive classes of TGDs that capture Datalog, the answer to the above question is negative, since checking whether a Datalog query is FO-rewritable is an undecidable problem. Actually, we know that a Datalog query is FO-rewritable iff it is bounded~\cite{AjGu94}, while the boundedness problem for Datalog is undecidable~\cite{GMSV93}.

The above question has been studied for OMQ languages based on Horn DLs, including $\mathcal{EL}$ and $\mathcal{ELI}$, which (up to a
certain normal form) are a special case of guarded TGDs~\cite{BiLW13,BiHLW16,LuSa17}.
More precisely, FO-rewritability is semantically characterized in terms of the existence of certain tree-shaped ABoxes, which in turn allows the authors to pinpoint the complexity of the problem by employing automata-based procedures.
As usual in the DL context, schemas consist only of unary and binary relations. However, in our setting we have to deal with relations of higher arity. This indicates that the techniques devised for checking the FO-rewritbility of DL-based OMQs cannot be directly applied to rule-based OMQs; this is further explained in Section~\ref{sec:semanticcharacterization}. Therefore, we develop new semantic characterizations and procedures that are significantly different from those for OMQs based on description logics.

Our analysis aims to develop specially tailored techniques that allow us to understand the problem of checking whether a (frontier-)guarded OMQ is FO-rewritable, and also to pinpoint its computational complexity. 
%To this end, as we discuss below, we follow two different approaches. 
Our plan of attack and results can be summarized as follows:

$\blacktriangleright$ We first focus on the simpler OMQ language based on guarded TGDs and atomic queries, and, in Section~\ref{sec:semanticcharacterization}, we provide a characterization of FO-rewritability that forms the basis for applying
tree automata techniques.

$\blacktriangleright$ We then exploit, in Section~\ref{sec:firstapproach}, standard two-way alternating parity tree automata. In particular, we reduce our problem to the problem of checking the finiteness of the language of an automaton. The reduction relies on a refined version of the characterization of FO-rewritability established in Section~\ref{sec:semanticcharacterization}. This provides a transparent solution to our problem based on standard tools, but it does not lead to an optimal result.

$\blacktriangleright$ Towards an optimal result, we use, in Section~\ref{sec:secondapproach}, a more sophisticated automata model, known as cost automata. This allows us to show that FO-rewritability for OMQs based on guarded TGDs and atomic queries is in \twoexp, and in \expo~for predicates of bounded arity.
Our application of cost automata is quite transparent, which, as above, relies on a refined version of the characterization of FO-rewritability established in Section~\ref{sec:semanticcharacterization}. However, the complexity analysis relies on an intricate result on the boundedness problem for a certain class of cost automata from~\cite{BeCCB15}.

$\blacktriangleright$ Finally, in Section~\ref{sec:treeification}, by using the results of Section~\ref{sec:secondapproach}, we obtain our main results.
We show that FO-rewritability is 2{\sc ExpTime}-complete for OMQs based on guarded TGDs and on frontier-guarded TGDs, no matter whether the actual queries are conjunctive queries, unions thereof, or the simple atomic queries. This remains true when the arity of the predicates is bounded by a constant, with the exception of guarded TGDs and atomic queries, for which the complexity
then drops to {\sc ExpTime}-complete.

In principle, the procedure based on tree automata also provides concrete FO-rewritings when they exist, but it is not tailored towards doing this in an efficient way. Efficiently constructing rewritings is beyond the scope of this work.

\section{Preliminaries}\label{sec:preliminaries}

\noindent
\textbf{Basics.} Let $\consts$, $\nulls$, and
$\vars$ be disjoint, countably infinite sets of \emph{constants},
\emph{(labeled) nulls}, and (regular) \emph{variables}, respectively. A \emph{schema} $\sche{S}$ is a finite set of relation symbols. The \emph{width} of $\sche{S}$, denoted $\width{\sche{S}}$, is the maximum arity among all relation symbols of $\sche{S}$. We write $R/n$ to denote that the relation symbol $R$ has arity $n \geq 0$.  A \emph{term} is either a constant, null, or variable. An \emph{atom} over $\sche{S}$ is an expression of the form $R(\ve{v})$, where $R \in \sche{S}$ is of arity $n \geq 0$ and $\ve{v}$ is an $n$-tuple of terms. A \emph{fact} is an atom whose arguments are constants.

\medskip
\noindent
\textbf{Databases.} An \emph{$\sche{S}$-instance} is a (possibly infinite) set of atoms over the schema $\sche{S}$ that contain only constants and nulls, while an \emph{$\sche{S}$-database} is a finite set of facts over $\sche{S}$. The \emph{active domain} of an instance $\inst{J}$, denoted $\adom{\inst{J}}$, consists of all terms occurring in $\inst{J}$. For $X \subseteq \adom{\inst{J}}$, we denote by $\inst{J}[X]$ the \emph{subinstance of $\inst{J}$ induced by $X$}, i.e., the set of all facts $R(\ve{a})$ with $\ve{a} \subseteq X$.
A \emph{tree decomposition} of an instance $\inst{J}$ is a tuple $\delta = \tup{\ca{T}, \tup{X_t}_{t \in T}}$, where $\ca{T} = \tup{T, E}$ is a (directed) tree with nodes $T$ and edges $E$, and $\tup{X_t}_{t \in T}$ is a collection of subsets of $\adom{\inst{J}}$, called \emph{bags}, such that
\begin{enumerate*}[label={(\roman*)}]
\item if $R(\ve{a}) \in \inst{J}$, then there is $v \in T$ such that $\ve{a} \subseteq X_v$, and

\item for all $a \in \adom{\inst{J}}$, the set
  $\set{v \in T \mid a \in X_v}$ induces a connected subtree of
  $\ca{T}$.
\end{enumerate*}
The \emph{width} of $\delta$ is the maximum size among all bags $X_v$ ($v \in T$) minus one. The \emph{tree-width} of $\inst{J}$, denoted $\tw{\inst{J}}$, is $\min\set{n \mid \text{there is a tree decomposition of width $n$ of \inst{J}}}$.

\medskip
\noindent
\textbf{Conjunctive queries.}
A \emph{conjunctive query} (CQ) over $\sche{S}$ is a first-order formula
of the form $q(\ve{x}) = \exists \ve{y}\, \varphi(\ve{x},\ve{y})$, where $\ve{x}$ and $\ve{y}$ are tuples of variables, and $\varphi$ is a conjunction of atoms $R_1(\ve{v}_1) \land \cdots \land R_m(\ve{v}_m)$ over $\sche{S}$ that mention variables from $\ve{x} \cup \ve{y}$ only. The variables $\ve{x}$ are the \emph{answer variables} of $q(\ve{x})$. If $\ve{x}$ is empty then $q$ is a \emph{Boolean CQ}. Let $\var{q}$ be the set of variables occurring in $q$.  As usual, the evaluation of CQs over instances is defined in terms of homomorphisms. A \emph{homomorphism} from $q$ to $\inst{J}$ is a mapping $h \colon \var{q} \rightarrow \adom{\inst{J}}$ such that $R_i(h(\ve{v}_i)) \in \inst{J}$ for each $1 \leq i \leq m$. We write $\inst{J} \models q(\ve{a})$ to indicate that there is such a homomorphism $h$ such that $h(\ve{x}) = \ve{a}$. The \emph{evaluation of $q(\ve{x})$ over $\inst{J}$}, denoted $q(\inst{J})$, is the set of all tuples $\ve{a}$ such that $\inst{J} \models q(\ve{a})$. A \emph{union of conjunctive queries} (UCQ) $q(\ve{x})$ over $\sche{S}$ is a disjunction $\bigvee_{i=1}^n q_i(\ve{x})$ of CQs over $\sche{S}$. The \emph{evaluation of $q(\ve{x})$ over $\inst{J}$}, denoted $q(\inst{J})$, is the set of tuples $\bigcup_{1 \leq i \leq n} q_i(\inst{J})$. We write $\inst{J} \models q(\ve{a})$ to indicate that $\inst{J} \models q_i(\ve{a})$ for some $1 \leq i \leq n$.
Let $\class{CQ}$ be the class of conjunctive queries, and $\class{UCQ}$ the class of UCQs. We also write $\class{AQ}_0$ for the class of {\em atomic queries} of the form $P()$, where $P$ is a $0$-ary predicate.

\medskip
\noindent
\textbf{Tuple-generating dependencies.} A \emph{tuple-generating
  dependency} (TGD) (a.k.a.~\emph{existential rule}) is a first-order
sentence of the form $\tau \colon \forall \ve{x},\ve{y}\,(\varphi(\ve{x},\ve{y}) \limpl \exists\ve{z}\,\psi(\ve{x},\ve{z}))$,
where $\varphi$ and $\psi$ are conjunctions of atoms that mention only
variables. For brevity, we write $\varphi(\ve{x},\ve{y}) \limpl \exists\ve{z}\,\psi(\ve{x},\ve{z})$, and use comma instead of $\land$ for conjoining atoms. We assume that each variable of $\ve{x}$ is mentioned in $\psi$. We call $\varphi$ and $\psi$ the \emph{body} and \emph{head}
of the TGD, respectively. The TGD $\tau$ is logically equivalent to the sentence $\forall \ve{x}\,(q_\varphi(\ve{x}) \limpl q_\psi(\ve{x}))$, where
$q_\varphi(\ve{x})$ and $q_\psi(\ve{x})$ are the CQs
$\exists \ve{y}\,\varphi(\ve{x},\ve{y})$ and $\exists\ve{z}\,\psi(\ve{x},\ve{z})$, respectively. Thus, an instance
$\inst{J}$ \emph{satisfies} $\tau$ if $q_\varphi(\inst{J}) \subseteq q_\psi(\inst{J})$. Also, $\inst{J}$ satisfies a set of TGDs $\ont{O}$, denoted $\inst{J} \models \ont{O}$, if $\inst{J}$ satisfies every $\tau \in \ont{O}$. Let $\class{TGD}$ be the class of finite sets of TGDs.

\medskip
\noindent
\textbf{Ontology-mediated queries.}
An \emph{ontology-mediated query} (OMQ) is a triple
$Q = \tup{\sche{S},\ont{O},q(\ve{x})}$, where $\sche{S}$ is a (non-empty) schema (the \emph{data schema}), $\ont{O}$ is a set of TGDs (the \emph{ontology}), and $q(\ve{x})$ is a UCQ over $\sche{S} \cup \sig{\ont{O}}$, where $\sig{\ont{O}}$ is the set of relation symbols in $\ont{O}$. Notice that the ontology $\ont{O}$ can introduce relations that are not in $\sche{S}$; this allows us to enrich the schema of $q(\ve{x})$. We include $\sche{S}$ in the specification of $Q$ to emphasize that $Q$ will be evaluated over $\sche{S}$-databases, even though $\ont{O}$ and $q(\ve{x})$ may use additional relation symbols.

The semantics of $Q$ is given in terms of certain answers. The
\emph{certain answers} to a UCQ $q(\ve{x})$ w.r.t.~an
$\sche{S}$-database $\db{D}$, and a set $\ont{O}$ of TGDs, is the set of
all tuples $\ve{a}$ of constants, where $|\ve{a}| = |\ve{x}|$, such that
$(\db{D}, \ont{O}) \models q(\ve{a})$, i.e., $\inst{J} \models
q(\ve{a})$ for every instance $\inst{J} \supseteq \db{D}$ that satisfies
$\ont{O}$. We write $\db{D} \models Q(\ve{a})$ if $\ve{a}$ is a certain
answer to $q$ w.r.t.~$\db{D}$ and $\ont{O}$. Moreover, we set $Q(\db{D})
= \set{\ve{a} \in \adom{\db{D}}^{|\ve{x}|} \mid \db{D} \models Q(\ve{a})}$.

\medskip
\noindent
\textbf{Ontology-mediated query languages.} We write $\tup{\class{C},\class{Q}}$ for the class of OMQs $\tup{\sche{S},\ont{O},q}$, where $\ont{O}$ falls in the class of TGDs $\class{C}$, and $q$ in the query language $\class{Q}$.
The evaluation problem for $\tup{\class{TGD},\class{UCQ}}$, i.e., given a query $Q \in \tup{\class{TGD},\class{UCQ}}$ with data schema $\sche{S}$, an $\sche{S}$-database $\db{D}$, and $\ve{a} \in \adom{\db{D}}^{|\ve{x}|}$, to decide whether $\db{D} \models Q(\ve{a})$, is undecidable; this holds even for $\tup{\class{TGD},\class{AQ}_0}$~\cite{CaGK13}.
Here we deal with one of the most paradigmatic decidable restrictions, i.e., \emph{guardedness}. A TGD is \emph{guarded} if it has a body atom, called \emph{guard}, that contains all the body variables. Let $\class{G}$ be the class of all finite sets of guarded TGDs.
A TGD $\tau$ is called \emph{frontier-guarded} if its body contains an
atom, called \emph{frontier-guard}, that contains the frontier of
$\tau$, i.e., the body variables that appear also in the head. We write
$\class{FG}$ for the class of all finite sets of frontier-guarded
TGDs. Roughly, the evaluation problem for $\tup{\class{G},\class{UCQ}}$ and
$\tup{\class{FG},\class{UCQ}}$ is decidable since $\class{G}$ and
$\class{FG}$ admit tree-like universal models~\cite{CaGK13}.

\medskip
\noindent
\textbf{First-order rewritability.}  A \emph{first-order (FO) query} over a
schema $\sche{S}$ is a (function-free) FO formula $\varphi(\ve{x})$, with $\ve{x}$ being its free variables, that uses only relations from $\sche{S}$.
The \emph{evaluation} of $\varphi$ over an $\sche{S}$-database $\db{D}$,
denoted $\varphi(\db{D})$, is the set of tuples
$\{\ve{a} \in \adom{\db{D}}^{|\ve{x}|} \mid \db{D} \models
\varphi(\ve{a})\}$; $\models$ denotes the standard notion of satisfaction for FO.
An OMQ $Q = \tup{\sche{S},\ont{O},q(\ve{x})}$ is \emph{FO-rewritable} if
there exists a (finite) FO query $\varphi_Q(\ve{x})$ over $\sche{S}$
that is \emph{equivalent} to $Q$, i.e., for every $\sche{S}$-database
$\db{D}$ it is the case that $Q(\db{D}) = \varphi_Q(\db{D})$. We call
$\varphi_Q(\ve{x})$ an \emph{FO-rewriting of $Q$}.
A fundamental task for an OMQ language $(\class{C},\class{Q})$, where
$\class{C}$ is a class of TGDs and $\class{Q}$ is a class of queries, is
deciding first-order rewritability:

\smallskip

\begin{center}
\fbox{\begin{tabular}{ll}
{\small PROBLEM} : & $\forew{\class{C}}{\class{Q}}$
\\
{\small INPUT} : & An OMQ $Q \in \tup{\class{C},\class{Q}}$.
\\
{\small QUESTION} : &  Is it the case that $Q$ is FO-rewritable?
\end{tabular}}
\end{center}

\medskip
\noindent
\textbf{First-order rewritability of $\boldsymbol{(\class{FG},\class{UCQ})}$-queries.} As shown by the following example, there exist $(\class{G},\class{CQ})$ queries (and thus, $(\class{FG},\class{UCQ})$ queries) that are not FO-rewritable.

\begin{example}
  \label{ex:cartwheel}
  Consider the OMQ $Q = \tup{\sche{S},\ont{O},q} \in \tup{\class{G},\class{CQ}}$, where $\sche{S} = \set{S/3,A/1,B/1}$, $\ont{O}$
  consists of
  \[
  \begin{array}{rcl}
    S(x,y,z), A(z) &\rightarrow&  R(x,z),\\
    S(x,y,z), R(x,z) &\rightarrow&  R(x,y),
  \end{array}
  \]
  and $q = \exists x,y,z\,(S(x,y,z) \land R(x,z) \land B(y))$.
  Intuitively, an FO-rewriting of $Q$ should check for the existence of a set of atoms $\{S(c,a_i,a_{i-1})\}_{1 \leq i \leq k}$, among others, for $k \geq 0$. However, since there is no upper bound for $k$, this cannot be done via a finite FO-query, and thus, $Q$ is not FO-rewritable. A proof that $Q$ is not FO-rewritable is given below. \hfill\markfull
\end{example}

On the other hand, there are (frontier-)guarded OMQs that are FO-rewritable; e.g., the OMQ obtained from the query $Q$ in Example~\ref{ex:cartwheel} by adding $A(z)$ to $q$ is FO-rewritable with $\exists x,y,z\,(S(x,y,z) \land B(y) \land A(z))$ being an FO-rewriting.

\section{Semantic Characterization}
\label{sec:semanticcharacterization}

We proceed to give a characterization of FO-rewritability of OMQs from
$\tup{\class{G},\class{AQ}_0}$ in terms of the existence of certain
tree-like databases. Our characterization is related to, but different
from characterizations used for OMQs based on DLs such as $\mathcal{EL}$ and $\mathcal{ELI}$~\cite{BiLW13,BiHLW16}.

The characterizations in~\cite{BiLW13,BiHLW16} essentially state that a unary OMQ $Q$ is FO-rewritable iff there is a bound $k$ such that, whenever the root of a tree-shaped database $\db{D}$ is returned as an answer to $Q$, then this is already true for the restriction of $\db{D}$ up to depth $k$.  The proof of the (contrapositive of the) ``only if'' direction uses a locality argument: if there is no such bound $k$, then this is witnessed by an infinite sequence of deeper and deeper tree databases that establish non-locality of $Q$. For guarded TGDs, we would have to replace tree-shaped databases with databases of bounded tree-width. However, increasing depth of tree decompositions does not correspond to increasing distance in the Gaifman graph, and thus, does not establish non-locality. We therefore depart from imposing a bound on the depth, and instead we impose a bound on the number of facts, as detailed below.

It is also interesting to note that, while it is implicit in~\cite{BiHLW16} that an OMQ based on $\mathcal{ELI}$ and CQs is FO-rewritable iff it is Gaifman local, there exists an OMQ from $\tup{\class{G},\class{CQ}}$ that is Gaifman local, but not FO-rewritable. Such an OMQ is the one obtained from the query $Q$ given in Example~\ref{ex:cartwheel}, by removing the existential quantification on the variable $x$ in the CQ $q$, i.e., converting $q$ into a unary CQ.

\begin{theorem}
  \label{pro:semanticmain}
  Consider an OMQ $Q \in \tup{\class{G},\class{AQ}_0}$ with data schema $\sche{S}$. The following are equivalent:
  \begin{enumerate}
  \item $Q$ is FO-rewritable.
  \item There is a $k \geq 0$ such that, for every
    $\sche{S}$-database $\db{D}$ of tree-width at most %$\max\set{0,\width{\sche{S}} - 1}$
    $\width{\sche{S}} - 1$, if $\db{D} \models Q$, then
    there is a $\db{D}' \subseteq \db{D}$ with at most $k$ facts such
    that $\db{D}' \models Q$.
  \end{enumerate}
\end{theorem}

For $\text{(1)} \Rightarrow \text{(2)}$ we exploit the fact that, if
$Q \in \tup{\class{G},\class{AQ}_0}$ is FO-rewritable, then it can be
expressed as a UCQ $q_Q$. This follows from the fact that OMQs from $\tup{\class{G},\class{AQ}_0}$ are preserved under homomorphisms~\cite{BCLW14}, and Rossman's Theorem stating that an FO query is preserved under homomorphisms over finite instances iff it is equivalent to a UCQ~\cite{Ro08}. It is then easy to show that (2) holds with $k$ being the size of the largest disjunct of the UCQ $q_Q$. For $\text{(2)} \Rightarrow \text{(1)}$, we use the fact that, if there is an $\sche{S}$-database $\db{D}$ that entails $Q$, then there exists one of tree-width at most 
%$\max\set{0,\width{\sche{S}} - 1}$ 
$\width{\sche{S}} - 1$ that entails $Q$, and can be mapped to $\db{D}$. The next example illustrates Theorem~\ref{pro:semanticmain}.

\begin{example}
 \label{ex:cartwheel2}
 Consider the OMQ $Q = \tup{\sche{S},\ont{O},P} \in \tup{\class{G},\class{AQ}_0}$, where $\sche{S} = \set{S/3,A/1,B/1}$, and $\ont{O}$ consists of the TGDs given in Example~\ref{ex:cartwheel} plus the guarded TGD
  \[
  S(x,y,z), R(x,z), B(y)\ \rightarrow\  P,
  \]
  which is essentially the CQ $q$ from Example~\ref{ex:cartwheel}. It is easy to verify that, for an arbitrary $k \geq 0$, the $\sche{S}$-database
  \[
  \db{D}_k = \{A(a_0), S(c,a_1,a_0),\ldots,S(c,a_{k-1},a_{k-2}),B(a_{k-1})\}
  \]
  of tree-width 
  %$\max\set{0,\width{\sche{S}} - 1} = 2$ 
  $\width{\sche{S}} - 1 = 2$ is such that
  $\db{D}_k \models Q$, but for every $\db{D}' \subset \db{D}_k$ with
  at most $k$ facts, $\db{D}' \not\models Q$. Thus, by Theorem~\ref{pro:semanticmain}, $Q$ is not FO-rewritable. \hfill\markfull
\end{example}

\section{Alternating Tree Automata Approach}
\label{sec:firstapproach}

In this section, we exploit the well-known algorithmic tool of {\em
two-way alternating parity tree automata} (2ATA) over finite trees of bounded degree (see, e.g.,~\cite{CGKV88}), and prove that $\forew{\class{G}}{\class{AQ}_0}$ can be solved in elementary time. Although this result is not optimal, our construction provides a transparent solution to $\forew{\class{G}}{\class{AQ}_0}$ based on standard tools.  This is in contrast with previous studies on closely related problems for guarded logics, in which all elementary bounds heavily rely on the use of intricate results on cost automata \cite{BOW14,BeCCB15}. We also apply such results later, but only in order to pinpoint the exact complexity of $\forew{\class{G}}{\class{AQ}_0}$.

The idea behind our solution to $\forew{\class{G}}{\class{AQ}_0}$ is,
given a query $Q \in (\class{G},\class{AQ}_0)$, to devise a 2ATA $\ata{B}_Q$ such that $Q$ is FO-rewritable iff the language accepted by $\ata{B}_Q$ is finite. This is a standard idea with roots in the study of the boundedness problem for {\em monadic Datalog} (see e.g., \cite{Va92}). In particular, our main result establishes the following:

\begin{theorem}
  \label{theorem:atafin} Let $Q \in (\class{G},\class{AQ}_0)$ with data schema $\sche{S}$. There is a 2ATA $\ata{B}_Q$ on trees of degree at most $2^{\width{\sche{S}}}$ such that $Q$ is FO-rewritable iff the language of $\ata{B}_Q$ is finite. The state set of $\ata{B}_Q$ is of double exponential size in $\width{\sche{S}}$, and of exponential size in $|\sche{S} \cup \sig{\ont{O}}|$. Furthermore, $\ata{B}_Q$ can be constructed in double exponential time in the size of $Q$.
\end{theorem}

As a corollary to Theorem \ref{theorem:atafin} we obtain the following result:

\begin{corollary}\label{coro:ata-approach}
$\forew{\class{G}}{\class{AQ}_0}$ is in \threeexp, and in~\twoexp~for predicates of bounded arity.
\end{corollary}

From Theorem~\ref{theorem:atafin}, to check whether a query $Q \in (\class{G},\class{AQ}_0)$ is FO-rewritable, it suffices to check that the language of $\ata{B}_Q$ is finite. The latter is done by first converting $\ata{B}_Q$ into a non-deterministic bottom-up tree automaton $\ata{B}'_Q$; see, e.g.,~\cite{Va98}. This incurs an exponential blowup, and thus, $\ata{B}'_Q$ has triple exponentially many states. We then check the finiteness of the language of $\ata{B}'_Q$ in polynomial time in the size of $\ata{B}'_Q$ by applying a standard reachability analysis; see~\cite{Va92}. For predicates of bounded arity, a similar argument as above provides a double exponential time upper bound.

In the rest of Section \ref{sec:firstapproach} we explain the proof of Theorem \ref{theorem:atafin}. The intuitive idea is to construct a 2ATA $\ata{B}_Q$ whose language corresponds to suitable encodings of databases $\db{D}$ of
bounded tree-width that ``minimally'' satisfy $Q$, i.e., $\db{D} \models Q$, but if we remove any atom from $\db{D}$, then $Q$ is no longer satisfied.

\medskip
\noindent
\textbf{A refined semantic characterization.} In order to apply an
approach based on 2ATA, it is essential to revisit the semantic characterization provided by Theorem~\ref{pro:semanticmain}. To this end, we need to introduce some auxiliary terminology.

Let $\db{D}$ be a database, and $\delta = \tup{\ca{T},\tup{X_v}_{v \in
    T}}$, where $\ca{T} = (T,E)$, a tree decomposition of $\db{D}$. An \emph{adornment} of the
pair $\tup{\db{D}, \delta}$ is a function $\eta \colon T \rightarrow
2^{\db{D}}$ such that $\eta(v) \subseteq \db{D}[X_v]$ for all $v \in T$,
and $\bigcup_{v \in T} \eta(v) = \db{D}$. Therefore, the pair
$\tup{\delta,\eta}$ can be viewed as a representation of the database
$\db{D}$ along with a tree decomposition of it. For the intended
characterization, it is important that this representation is free of
redundancies, formalized as follows. We say that $\delta$ is \emph{$\eta$-simple}
if $|\eta(v)| \leq 1$ for all $v \in T$, and non-empty $\eta$-labels are
unique, that is, $\eta(v) \neq \eta(w)$ for all distinct $v,w \in T$
with $\eta(v)$ and $\eta(w)$ non-empty. Nodes $v \in T$ with
$\eta(v)$ empty, called \emph{white} from now on, are required since we
might not have a (unique!) fact available to label them. Note, though,
that white nodes $v$ are still associated with a non-empty set of
constants from $\db{D}$ via $X_v$. All other nodes are called \emph{black}.
While $\delta$ being $\eta$-simple avoids redundancies that are due to a fact
occurring in the label of multiple black nodes, additional redundancies
may arise from the inflationary use of white nodes. We say that a node
$v \in T$ is \emph{$\eta$-well-colored} if it is black, or it has at
least two successors and all its successors are $\eta$-well-colored. We say that $\delta$ is \emph{$\eta$-well-colored} if every node in $T$ is $\eta$-well-colored. For example, $\delta$ is not $\eta$-well-colored if it has a white leaf, or if it has a white node and its single successor is also white. Informally, requiring $\delta$ to be $\eta$-well-colored makes it impossible to blow up the tree by introducing white nodes without introducing black nodes.
For $i \in \{1,2\}$, let $\db{D}_i$ be a database, $\delta_i$ a tree
decomposition of $\db{D}_i$, and $\eta_i$ an adornment of
$\tup{\db{D}_i,\delta_i}$. We say that $\tup{\db{D}_1,\delta_1,\eta_1}$ and
$\tup{\db{D}_2,\delta_2,\eta_2}$ are \emph{isomorphic} if the latter can be
obtained from the former by consistenly renaming constants in
$\db{D}_1$ and tree nodes in $\delta_1$.

We are now ready to revisit the characterization of
FO-rewritability for OMQs from $\tup{\class{G},\class{AQ}_0}$ given in
Theorem~\ref{pro:semanticmain}.

\begin{theorem}
  \label{pro:semanticrefined}
  Consider an OMQ $Q \in \tup{\class{G},\class{AQ}_0}$ with data schema
  $\sche{S}$. The following are equivalent:
  \begin{enumerate}
  \item Condition 2 from Theorem~\ref{pro:semanticmain} is satisfied.

  \item There are finitely many non-isomorphic triples
    $\tup{\db{D},\delta, \eta}$, where $\db{D}$ is an
    $\sche{S}$-database, $\delta$ a tree decomposition of $\db{D}$ of
    width at most
    %$\max\set{0,\width{\sche{S}} - 1}$,
    $\width{\sche{S}} - 1$, and $\eta$ an adornment of $\tup{\db{D},\delta}$, such that\label{lem:semanticrefined:2}
  \begin{enumerate}

  \item $\delta$ is $\eta$-simple and $\eta$-well-colored,

  \item $\db{D} \models Q$, and

  \item for every $\alpha \in \db{D}$, it is the case that
    $\db{D} \setminus \set{\alpha} \not\models Q$.
  \end{enumerate}
  \end{enumerate}
\end{theorem}

\noindent
\textbf{Devising automata.} We proceed to discuss how the 2ATA announced in Theorem~\ref{theorem:atafin} is constructed.
Consider an OMQ $Q = \tup{\sche{S},\ont{O},P}$ from $\tup{\class{G},\class{AQ}_0}$. Our goal is to devise an automaton $\ata{B}_Q$ whose language is finite iff Condition 2 from Theorem~\ref{pro:semanticrefined} is satisfied.
By Theorems~\ref{pro:semanticmain}
and~\ref{pro:semanticrefined}, $Q$ is then FO-rewritable iff the language of
$\ata{B}_Q$ is finite.

The 2ATA $\ata{B}_Q$ will be the intersection of several automata
that check the properties stated in~\cref{lem:semanticrefined:2} of~\Cref{pro:semanticrefined}. But first we need to say a few words about tree encodings.
Let $\Gamma$ be a finite alphabet, and let $(\mathbb{N} \setminus \set{0})^\ast$ denote the set of all finite words of positive integers, including the empty word. A \emph{finite $\Gamma$-labeled tree} is a partial function $t \colon (\mathbb{N} \setminus \set{0})^\ast \rightarrow \Gamma$ such that the domain of $t$ is finite and prefix-closed. Moreover, if $v \cdot i$ belongs to the domain of $t$, then $v \cdot (i-1)$ also belongs to the domain of $t$. In fact, the elements in the domain of $t$ identify the nodes of the tree.
It can be shown that an $\sche{S}$-database $\db{D}$, a tree
decomposition $\delta$ of $\db{D}$ of width $w-1$, and an adornment $\eta$
of $\tup{\db{D},\delta}$, can be encoded as a $\Gamma_{\sche{S},w}$-labeled tree $t$ of degree at most $2^{w}$, where $\Gamma_{\sche{S},w}$ is
an alphabet of size double exponential in $w$ and exponential in $\sche{S}$, such that each node of $\delta$ corresponds to exactly one node of $t$ and vice versa. Although every $\db{D}$ can be encoded into a
$\Gamma_{\sche{S},w}$-labeled tree $t$, the converse is not true in
general.  However, it is possible to define certain syntactic
consistency conditions such that every \emph{consistent} $\Gamma_{\sche{S},w}$-labeled $t$ can be decoded into an $\sche{S}$-database, denoted $\dec{t}$, whose tree-width is at most $w$. We are going to abbreviate the alphabet
%$\Gamma_{\sche{S},\max\set{0,\width{\sche{S}}-1}}$
$\Gamma_{\sche{S},\width{\sche{S}}}$ by $\Gamma_{\sche{S}}$. %Consistency of $\Gamma_{\sche{S}}$-labeled trees can easily be checked using tree automata.

\begin{lemma}
  \label{lem:ataconsistency}
  There is a 2ATA $\ata{C}_{\sche{S}}$ that accepts a
  $\Gamma_{\sche{S}}$-labeled tree $t$ iff $t$ is consistent. The number
  of states of $\ata{C}_{\sche{S}}$ is constant. $\ata{C}_{\sche{S}}$ can be constructed in polynomial time in the size of $\Gamma_{\sche{S}}$.
\end{lemma}

Since a $\Gamma_{\sche{S}}$-labeled tree incorporates the information
about an adornment, the notions of being well-colored and simple can
be naturally defined for $\Gamma_{\sche{S}}$-labeled trees. Then:
%This in turn allows us to check these conditions via tree automata:

\begin{lemma}
  \label{lem:atarectwellc}
  There is a 2ATA $\ata{R}_{\sche{S}}$ that accepts a consistent
  $\Gamma_{\sche{S}}$-labeled tree iff it is well-colored and
  simple. The number of states of $\ata{R}_{\sche{S}}$ is exponential in
  $\width{\sche{S}}$ and linear in $|\sche{S}|$. $\ata{R}_{\sche{S}}$
  can be constructed in polynomial time in the size of
  $\Gamma_{\sche{S}}$.
\end{lemma}

Concerning property 2(b) of~\Cref{pro:semanticrefined}, we can devise a
2ATA that accepts those trees whose decoding satisfies $Q$:

\begin{lemma}
  \label{lem:atasat}
  There is a 2ATA $\ata{A}_Q$ that accepts a consistent
  $\Gamma_{\sche{S}}$-labeled tree iff $\dec{t} \models Q$.  The number
  of states of $\ata{A}_Q$ is exponential in $\width{\sche{S}}$ and
  linear in $|\sche{S} \cup \sig{\ont{O}}|$.  $\ata{A}_Q$ can be
  constructed in double exponential time in the size of $Q$.
\end{lemma}

The crucial task is to check condition 2(c)
of~\Cref{pro:semanticrefined}, which states the key minimality
criterion. Unfortunately, this involves an extra exponential blowup:

\begin{lemma}
  \label{lem:atamin}
  There is a 2ATA $\ata{M}_Q$ that accepts a consistent $\Gamma_{\sche{S}}$-labeled tree $t$ iff $\dec{t} \setminus \set{\alpha} \not\models Q$ for all $\alpha \in \dec{t}$. The state set of $\ata{M}_Q$ is of double exponential size in $\width{\sche{S}}$, and of exponential size
  in $|\sche{S} \cup \sig{\ont{O}}|$. Furthermore, $\ata{M}_Q$ can be
  constructed in double exponential time in the size of $Q$.
\end{lemma}

Let us briefly explain how $\ata{M}_Q$ is constructed. This will expose the source of the extra exponential blowup, which prevents us from obtaining an optimal complexity upper bound for $\forew{\class{G}}{\class{AQ}_0}$.
We first construct a 2ATA $\ata{D}_Q$ that runs on
$\Lambda_{\sche{S}}$-labeled trees, where $\Lambda_{\sche{S}}$ is an
alphabet that extends $\Gamma_{\sche{S}}$ with auxiliary symbols that
allow us to tag some facts in the input tree. In particular, $\ata{D}_Q$
accepts a tree $t$ iff $t$ is consistent, there is at least one tagged
fact, and $\dec{t}^{-} \models Q$ where $\dec{t}^{-}$ is obtained from
$\dec{t}$ by removing the tagged facts.
Having $\ata{D}_Q$ in place, we can then construct a 2ATA $\exists\ata{D}_Q$ that accepts a $\Gamma_{\sche{S}}$-labeled tree $t$ if there is a way to tag some of its facts so as to obtain a $\Lambda_{\sche{S}}$-labeled tree $t'$ with $\dec{t'}^{-} \models Q$. This is achieved by applying the projection operator on $\ata{D}_Q$. Since for 2ATAs projection involves an exponential blowup and $\ata{D}_Q$ already has exponentially many states, $\exists \ata{D}_Q$ has double exponentially many.
It should be clear now that $\ata{M}_Q$ is the complement of $\exists\ata{D}_Q$, and we recall that complementation of 2ATAs can be
done in polynomial time.

The desired automaton $\ata{B}_Q$ is obtained by intersecting the 2ATAs in~\Cref{lem:ataconsistency,lem:atarectwellc,lem:atasat,lem:atamin}. Since the intersection of 2ATA is feasible in polynomial time, $\ata{B}_Q$ can be constructed in double exponential time in the size of $Q$.

\section{Cost Automata Approach}\label{sec:secondapproach}

We proceed to study $\forew{\class{G}}{\class{AQ}_0}$
%a non-standard automata
using the more sophisticated model of cost automata.  This allows us to
improve the complexity of the problem obtained in Corollary~\ref{coro:ata-approach} as follows:

\begin{theorem}\label{the:cost-automata-approach}
  $\forew{\class{G}}{\class{AQ}_0}$ is in \twoexp, and in \expo~for predicates of bounded arity.
\end{theorem}

As in the previous approach, we develop a semantic characterization that relies on a minimality criterion for trees accepted by cost automata. The extra features provided by cost automata allow us to deal with such a minimality criterion in a more efficient way than standard 2ATA.
%, and thus, improve the complexity.
While our application of cost automata is transparent, the complexity analysis relies on an intricate result on the boundedness problem for a certain class of cost automata from~\cite{BeCCB15}.
Before we proceed further, let us provide a brief overview of the cost automata model that we are going to use. 
%We refer the reader to the appendix for a more detailed description.

\medskip
\noindent
\textbf{Cost automata models.} Cost automata extend traditional automata (on words, trees, etc.) by providing counters that can be manipulated at each transition. Instead of assigning a Boolean value to each input structure (indicating whether the input is accepted or not), these automata assign a value from $\mathbb{N}_\infty = \mathbb{N} \cup \set{\infty}$ to each
input.

Here, we focus on cost automata that work on finite trees of unbounded degree, and allow for two-way movements; in fact, the automata that we need are those that extend 2ATA over finite trees with a \emph{single} counter. The operation of such an automaton $\ata{A}$ on each input $t$ will be viewed as a two-player \emph{cost game} $\ca{G}(\ata{A},t)$ between players Eve and Adam. Recall that the acceptance of an input tree for a conventional 2ATA can be formalized via a two-player game as well. However, instead of the parity acceptance condition for 2ATA, plays in the cost game between Eve and Adam will be assigned costs, and the cost automaton specifies via an \emph{objective} whether Eve's goal is to minimize or maximize that cost. In case of a minimizing (resp., maximizing) objective, a strategy $\xi$ of Eve in the cost game $\ca{G}(\ata{A},t)$ is \emph{$n$-winning} if any play of Adam consistent with $\xi$ has cost at most $n$ (resp., at least $n$). Given an input tree $t$, one then defines the \emph{value of $t$ in $\ata{A}$} as
\begin{align*}
\val{\ata{A}}(t) = \mathrm{op}\set{n \mid \text{Eve has an $n$-winning strategy in $\ca{G}(\ata{A},t)$}},
\end{align*}
where $\mathrm{op} = \inf$ (resp., $\mathrm{op} = \sup$) in case
Eve's objective is to minimize (resp., maximize). Therefore, $\val{\ata{A}}$ defines a function from the domain of input trees to $\mathbb{N}_\infty$. We call functions of that type \emph{cost functions}.
A key property of such functions is boundedness. We say that
$\val{\ata{A}}$ is {\em bounded} if there exists an $n \in \mathbb{N}$ such
that $\val{\ata{A}}(t) \leq n$ for every input tree $t$.

%The objects of interest in the context of cost automata are boundedness
%properties of the cost functions they are able to define. That is,
%we will be interested in the questionwhether the cost function $\val{\ata{A}}$
%defined by an automaton $\ata{A}$ is \emph{bounded}---we say that a
%function $f$ is \emph{bounded} over a domain $D$ (in our case finite
%labeled trees), if there is an $n \in \mathbb{N}$ such that
%$f(t) \leq n$ for all $t \in D$. The \emph{boundedness question for
%  $\ata{A}$} asks whether the function $\val{\ata{A}}$ is bounded over
%all input structures in its domain.

We employ automata with a single counter, where Eve's
objective is to minimize the cost, while satisfying the parity condition. Such an automaton is known in the literature as \emph{$\dist \land \parity$-automaton}~\cite{BeCCB15}. To navigate in the tree, it may use the directions $\set{0, \updownarrow}$, where $0$ indicates that the automaton should stay in the current node, and $\updownarrow$ means that the automaton may move to an arbitrary neighboring node, including the parent. For this type of automaton, we can decide whether its cost function is bounded~\cite{BeCCB15,CoF16}. As usual, $\size{\ata{A}}$ denotes the size $\ata{A}$. Then:

%:\todo{Include a footnote clarifying where this result is actually proved, i.e., some of Colcombet's recent work}

\begin{theorem}
  \label{thm:boundednessdist}
  There is a polynomial $f$ such that, for every
  $\dist \land \parity$-automaton $\ata{A}$ using priorities $\set{0,1}$
  for the parity acceptance condition, the boundedness for
  $\val{\ata{A}}$ is decidable in time $\size{\ata{A}}^{f(m)}$, where
  $m$ is the number of states of $\ata{A}$.
\end{theorem}

Our goal is to reduce $\forew{\class{G}}{\class{AQ}_0}$ to the
boundedness problem for $\dist \land \parity$-automata.
%, which in turn will allow us to show that $\forew{\class{G}}{\class{AQ}_0}$ is in \twoexp.

%In the following, we are going to devise a decision procedure for $\forew{\class{G}}{\class{AQ}_0}$, which runs in double-exponential time, by  relying on $\dist \land \parity$-automata. In fact, out intention is to reduce our problem to the boundedness problem for $\dist \land \parity$-automata.

%In the following, we are going to devise a semantic characterization for
%$\forew{\class{G}}{\class{BAQ}_0}$ that relies
%on~\Cref{lem:semanticmain} and which will be convenient to decide using
%$\dist \land \parity$-automata. Before that, we need some preliminary
%concepts.

\medskip
\noindent
\textbf{A refined semantic characterization.}
%With the aim of devising such a reduction,
We first need to revisit the semantic characterization provided by Theorem~\ref{pro:semanticmain}.
%To this end, we need to introduce some auxiliary terminology.

Consider an $\sche{S}$-database $\db{D}$, and a query $Q = \tup{\sche{S},\ont{O},P} \in \tup{\class{G},\class{AQ}_0}$. Let $k_Q = |\sche{S} \cup \sig{\ont{O}}|\cdot w^w$, where $w = \width{\sche{S} \cup \sig{\ont{O}}}$.
A \emph{derivation tree} for $\db{D}$ and $Q$ is a labeled $k_Q$-ary tree $\ca{T}$, with $\eta$ being a node labeling function that assigns facts $R(\ve{a})$, where $R \in \sche{S} \cup \sig{\ont{O}}$ and $\ve{a} \subseteq \adom{\db{D}}$, to its nodes, that satisfies the following conditions:
\begin{enumerate}
  \item For the root node $v$ of $\ca{T}$, $\eta(v) = P$.

  \item For each leaf node $v$ of $\ca{T}$, $\eta(v) \in \db{D}$.

  \item For each non-leaf node $v$ of $\ca{T}$, with $u_1,\ldots,u_k$
    being its children,
    %a) there exists an $i \in \set{1,\ldots,k}$ such that $\adom{\set{\eta(u_1),\ldots,\eta(u_k)}} \subseteq \adom{\set{\eta(u_i)}}$,
    %and (b)
    $(\set{\eta(u_1),\ldots,\eta(u_k)},\ont{O}) \models \eta(v)$.
\end{enumerate}
Roughly, $\ca{T}$ describes how the $0$-ary predicate $P$ can be entailed from $\db{D}$ and $\ont{O}$. In fact, it is easy to show that $\db{D} \models Q$ iff there is a derivation tree for $\db{D}$ and $Q$.
The \emph{height} of $\ca{T}$, denoted $\height{\ca{T}}$, is the maximum length of a branch in $\ca{T}$, i.e., of a path from the root to a leaf node. Assuming that $\db{D} \models Q$, the \emph{cost of $\db{D}$ w.r.t.~$Q$}, denoted $\cost{\db{D}}{Q}$, is defined as
%$\min\set{\height{\ca{T}} \mid \ca{T} \text{ is a derivation tree for } \db{D} \text{ and } Q}$,
\[
\min\set{\height{\ca{T}} \mid \text{$\ca{T}$ is a derivation tree for $\db{D}$ and $Q$}},
\]
while the \emph{cost} of $Q$, denoted \costq{Q}, is defined as
\begin{align*}
    \sup \set{\cost{\db{D}}{Q} \mid\ &\text{$\db{D} \models Q$, $\db{D}$ is an $\sche{S}$-database} \\
                                                      &\text{with $\tw{\db{D}} \leq \max\set{0, \width{\sche{S}}-1}$}}.
\end{align*}
In other words, the cost of $Q$ is the {\em least upper bound} of the
height over all derivation trees for all $\sche{S}$-databases $\db{D}$
of width at most $\max\set{0, \width{\sche{S}}-1}$ such that
$\db{D} \models Q$. If there is no such a database, then the cost of $Q$
is zero since $\sup \emptyset = 0$. Actually,
$\costq{Q} = 0$ indicates that $Q$ is unsatisfiable,
%, i.e., there is no $\sche{S}$-database that entails $Q$,
which in turn means that $Q$ is trivially FO-rewritable.

Having the notion of the cost of an OMQ from $\tup{\class{G},\class{AQ}_0}$ in place, it should not be difficult to see how we can refine the semantic
characterization provided by Theorem~\ref{pro:semanticmain}.

\begin{theorem}
  \label{pro:semanticcost}
  Consider an OMQ $Q \in \tup{\class{G},\class{AQ}_0}$ with data schema $\sche{S}$. The following are equivalent:
  \begin{enumerate}
  %\item There is a $k \geq 0$ such that, for every $\sche{S}$-database $\db{D}$ of tree-width at most $\max\set{0,\width{\sche{S}} - 1}$, if $\db{D} \models Q$, then there is a $\db{D}' \subseteq \db{D}$ with at most $k$ facts such that $\db{D}' \models Q$.
  \item Condition 2 from Theorem~\ref{pro:semanticmain} is satisfied.
  \item $\costq{Q}$ is finite.
  \end{enumerate}
\end{theorem}

\noindent
\textbf{Devising automata.} We briefly describe how we can use cost automata %together with the refined semantic characterization given in Theorem~\ref{pro:semanticcost},
in order to devise an algorithm for $\forew{\class{G}}{\class{AQ}_0}$ that runs in double exponential time.

Consider an OMQ $Q = \tup{\sche{S},\ont{O},P} \in \tup{\class{G},\class{AQ}_0}$. Our goal is to devise a $\dist \land \parity$-automaton $\ata{B}_Q$ such that the cost function $\val{\ata{B}_Q}$ is bounded iff $\costq{Q}$ is finite. Therefore, by Theorems~\ref{pro:semanticmain} and~\ref{pro:semanticcost}, to check whether $Q$ is FO-rewritable we simply need to check if $\val{\ata{B}_Q}$ is bounded, which, by Theorem~\ref{thm:boundednessdist}, can be done in exponential time in the size of $\ata{B}_Q$.
The input trees to our automata will be over the same alphabet $\Gamma_{\sche{S}}$ that is used to encode
tree-like $\sche{S}$-databases in~\Cref{sec:firstapproach}.
Recall that for a $\dist \land \parity$-automaton $\ata{A}$, the cost function $\val{\ata{A}}$ is bounded over a certain class $\ca{C}$ of trees if there is an $n \in \mathbb{N}$ such that $\val{\ata{A}}(t) \leq n$ for every input tree $t \in \ca{C}$. Then:
%
%The following lemma is the main ingredient towards a double exponential time procedure for $\forew{\class{G}}{\class{AQ}_0}$:

\begin{lemma}
  \label{lem:costata}
  There is a $\dist \land \parity$-automaton $\ata{H}_Q$ such that   $\val{\ata{H}_Q}$ is bounded over consistent $\Gamma_{\sche{S}}$-labeled trees iff $\costq{Q}$ is finite. The number of states of $\ata{H}_Q$ is exponential in $\width{\sche{S}}$, and polynomial in $|\sche{S} \cup \sig{\ont{O}}|$. Moreover, $\ata{H}_Q$ can be constructed in double exponential time in the size of $Q$.
\end{lemma}

The automaton $\ata{H}_Q$ is built in such a way that, on an input tree $t$, Eve has an $n$-winning strategy in $\ca{G}(\ata{H}_Q,t)$ iff there is a derivation tree for $\dec{t}$ and $Q$ of height at most $n$.
%
%Roughly, $\ata{H}_Q$ is built in such a way that an $n$-winning strategy
%for Eve gives rise to a witnessing derivation tree for the input
%$\Gamma_{\sche{S}}$-labeled tree $t$ and $Q$ that has height at most
%$n$. More formally, Eve has an $n$-winning strategy in
%$\ca{G}(\ata{H}_Q,t)$ iff there is a derivation tree for $\dec{t}$ and $Q$
%of height at most $n$.
%
Thus, Eve tries to construct derivation trees of minimal height. The counter is used to count the height of the derivation tree.

Having this automaton in place, we can now complete the proof of Theorem~\ref{the:cost-automata-approach}. The desired $\dist \land \parity$-automaton $\ata{B}_Q$ is defined as $\ata{C}'_{\sche{S}} \cap \ata{H}_Q$, where $\ata{C}'_{\sche{S}}$ is similar to the 2ATA $\ata{C}_\sche{S}$ (in Lemma~\ref{lem:ataconsistency}) that checks for consistency of $\Gamma_{\sche{S}}$-labeled trees of bounded degree.
%(recall that here we work with trees of unbounded degree).
Notice that $\ata{C}'_{\sche{S}}$ is essentially a $\dist \land \parity$-automaton that assigns zero (resp.,~$\infty$) to input trees that are consistent (resp.,~inconsistent), and thus, $\ata{C}'_{\sche{S}} \cap \ata{H}_Q$ is well-defined.
%without counters.
Since the intersection of $\dist \land \parity$-automata is feasible in polynomial time, Lemma~\ref{lem:ataconsistency} and Lemma~\ref{lem:costata} imply that $\ata{B}_Q$ has exponentially many states, and it can be constructed in double exponential time. Lemma~\ref{lem:costata} implies also that $\val{\ata{B}_Q}$ is bounded iff $\costq{Q}$ is finite.
It remains to show that the boundedness of $\val{\ata{B}_Q}$ can be checked in double exponential time. By \Cref{thm:boundednessdist}, there is a polynomial $f$ such that the latter task can be carried out in time $\size{\ata{B}_Q}^{f(m)}$, where $m$ is the number of states of $\ata{B}_Q$, and the claim follows.
For predicates of bounded arity, a similar complexity analysis as above shows a single exponential time upper bound.

\section{Frontier-Guarded OMQs}\label{sec:treeification}

The goal of this section is to show the following result:

\begin{theorem}\label{the:main-result}
It holds that:
\begin{itemize}
\item $\forew{\class{FG}}{\class{Q}}$ is~\twoexp-complete, for each $\class{Q} \in \{\class{UCQ},\class{CQ},\class{AQ}_0\}$, even for predicates of bounded arity.

\item $\forew{\class{G}}{\class{Q}}$ is~\twoexp-complete, for each $\class{Q} \in \{\class{UCQ},\class{CQ}\}$, even for predicates of bounded arity.

\item $\forew{\class{G}}{\class{AQ}_0}$ is~\twoexp-complete. Moreover, for predicates of bounded arity it is \expo-complete.
\end{itemize}
\end{theorem}

\noindent
\textbf{Lower bounds.}
The \twoexp-hardness in the first and the second items is inherited from~\cite{BiHLW16}, where it is shown that deciding FO-rewritability for OMQs based on $\mathcal{ELI}$ and CQs is \twoexp-hard.
For the \twoexp-hardness in the third item we exploit the fact that containment for OMQs from $(\class{G},\class{AQ}_0)$ is \twoexp-hard, even if the right-hand side query is FO-rewritable; this is implicit in~\cite{BaRV14}.
Finally, the \expo-hardness in the third item is inherited from~\cite{BiLW13}, where it is shown that deciding FO-rewritability for OMQs based on $\mathcal{EL}$ and atomic queries is \expo-hard.

\medskip
\noindent
\textbf{Upper bounds.} The fact that for predicates of bounded arity $\forew{\class{G}}{\class{AQ}_0}$ is in \expo~is obtained from Theorem~\ref{the:cost-automata-approach}.
It remains to show that $\forew{\class{FG}}{\class{UCQ}}$ is in~\twoexp.
We reduce $\forew{\class{FG}}{\class{UCQ}}$ in polynomial time to $\forew{\class{FG}}{\class{AQ}_0}$, and then show that the latter is in \twoexp. This reduction relies on a construction from~\cite{BiHLW16}, which allows us to reduce $\forew{\class{FG}}{\class{UCQ}}$ to $\forew{\class{FG}}{\class{UBCQ}}$ with $\class{UBCQ}$ being the class of union of Boolean CQs, and the fact that a Boolean CQ can be seen as a frontier-guarded TGD.
To show that $\forew{\class{FG}}{\class{AQ}_0}$ is in \twoexp, we reduce it to $\forew{\class{G}}{\class{AQ}_0}$, and then apply Theorem~\ref{the:cost-automata-approach}. This relies on {\em treeification}, and is inspired by a translation of guarded negation fixed-point sentences into guarded fixed-point sentences~\cite{BaCS15}.
%
%Roughly, a frontier-guarded TGD is converted into a set of guarded TGDs by treeifying the body of the former. 
Our reduction may give rise to exponentially many guarded TGDs, but the arity is increased only polynomially. Since the procedure for $\forew{\class{G}}{\class{AQ}_0}$ provided by Theorem~\ref{the:cost-automata-approach} is double exponential only in the arity of the schema the claim follows.

\section{Future Work}\label{sec:cinclusion}

The procedure based on 2ATA provides an FO-rewriting in case the input OMQ admits one, but it is not tailored towards doing this in an efficient way. Our next step is to exploit the techniques developed in this work for devising practically efficient algorithms for constructing FO-rewritings.

%\medskip

%\noindent \textbf{Acknowledgements.}

\section*{Acknowledgements}
Barcel\'o is funded by the Millennium Institute for Foundational Research on Data and Fondecyt grant 1170109.
Berger is funded by the FWF project W1255-N23 and a DOC fellowship of the Austrian Academy of Sciences.
Lutz is funded by the ERC grant 647289 ``CODA''.
Pieris is funded by the EPSRC programme grant EP/M025268/ ``VADA''.

%%% Local Variables:
%%% fill-column: 72
%%% TeX-PDF-mode: t
%%% TeX-debug-bad-boxes: t
%%% TeX-master: "ijcai18.tex"
%%% TeX-parse-self: t
%%% TeX-auto-save: t
%%% reftex-plug-into-AUCTeX: t
%%% End: 

\section*{\LARGE Appendix}
\appendix

\section{Proofs for~\Cref{sec:semanticcharacterization}}

\subsection{Proof of~\Cref{pro:semanticmain}}

Let us first cite an important lemma that will be used in the proof
of~\Cref{pro:semanticmain} below:
\begin{lemma}
  \label{lem:treemodel}
  Let $Q$ be an OMQ from $\tup{\class{G},\class{AQ}_0}$ with data schema
  $\sche{S}$ and consider an $\sche{S}$-database $\db{D}$. If
  $\db{D} \models Q$ then there is an $\sche{S}$-database $\db{D}^\ast$
  of tree-width at most $\max\set{0, \width{\sche{S}} - 1}$ such that
  \begin{enumerate}
    \item $\db{D}^\ast \models Q$ and
    \item there is a homomorphism from $\db{D}^\ast$ to $\db{D}$.
  \end{enumerate}
\end{lemma}
\Cref{lem:treemodel} can be proved using the notion of \emph{guarded
  unraveling} and applying the compactness theorem (an almost verbatim
result can be found in~\cite{ourpods}).

\smallskip
\noindent
\textit{Proof of~\Cref{pro:semanticmain}.} Assume first that $Q$ is
FO-rewritable. Then there is a first-order sentence $\varphi_Q$ which is
equivalent over all $\sche{S}$-databases to $Q$. Notice that $Q$ is
closed under homomorphisms, hence so is $\varphi_Q$. By Rossman's
theorem~\cite{Ro08}, we thus know that $\varphi_Q$ must be
equivalent to a (Boolean) UCQ $q_Q = \bigvee_{i = 1}^n p_i$. Now we let
$k = \max\set{|p_i| \colon i = 1,\ldots,n}$, where $|p_i|$ denotes
the number of atoms in $p_i$. We claim that $k$ is the bound we are
looking for in condition 2. Indeed, if $\db{D} \models Q$, for a database
of tree-width at most $\max\set{0,\width{\sche{S}}-1}$, then also
$\db{D} \models q_Q$ and so there is a homomorphism $h$ that maps some
$p_i$ to $\db{D}$. The image of $p_i$ under $h$ is a database of size at
most $k$ that satisfies $q_Q$, i.e., $h(p_i) \models q_Q$. Since $Q$ is
equivalent to $q_Q$, we infer $h(p_i) \models Q$, as required.

Suppose now that there is a $k \geq 0$ such that, for every
$\sche{S}$-database $\db{D}$ of tree-width at most
$\max\set{0, \width{\sche{S}}- 1}$, if $\db{D} \models Q$, then there is
a $\db{D}' \subseteq \db{D}$ with at most $k$ facts such that
$\db{D}' \models Q$. Let $\ca{S}$ be the class of all
$\sche{S}$-databases $\db{D}$ such that
\begin{enumerate*}[label={(\roman*)}]
  \item $\db{D}$ contains at most $k$ facts and
  \item $\db{D} \models Q$.
\end{enumerate*}
Consider $\ca{S}$ factorized modulo isomorphism.\footnote{Two databases
  are \emph{isomorphic}, if there is a bijective homomorphism between
  them.} Notice that $\ca{S}$ is thus finite. We claim that
$q_Q = \bigvee \ca{S}$ (here we consider the databases in
$\ca{S}$ as CQs) is a UCQ equivalent to $Q$ (and thus an FO-rewriting of
$Q$).

To see this, suppose first that
$\db{D} \models Q$ for some $\sche{S}$-database
$\db{D}$. By~\Cref{lem:treemodel}, there is an $\sche{S}$-database
$\db{D}^\ast$ of tree-width at most $\max{0, \width{\sche{S}}-1}$ such
that $\db{D}^\ast \models Q$. By assumption, there is a
$\db{D}_0 \subseteq \db{D}^\ast$ of at most $k$ facts such that
$\db{D}_0 \models Q$. It follows that some isomorphic representative of
$\db{D}_0$ is contained in $\ca{S}$. Therefore, $\db{D}_0 \models q_Q$
and, since $\db{D}_0 \subseteq \db{D}$, also $\db{D} \models q_Q$.

Suppose now that $\db{D} \models q_Q$. Then there is some
$p \in \ca{S}$ such that $\db{D} \models p$. Hence there is a
homomorphism $h$ that maps $p$ to $\db{D}$. Recall that $p$ (viewed as
an $\sche{S}$-database) also satisfies $Q$ by construction of
$\ca{S}$. Since $Q$ is closed under homomorphisms, we must also have
$\db{D} \models Q$, and the claim follows.
\hfill$\square$

\section{Proofs for~\Cref{sec:firstapproach}}

\subsection{Proof of~\Cref{pro:semanticrefined}}

Let $w = \max\set{0, \width{\sche{S}} -1}$. We are first going
to prove some auxiliary statements.

\medskip
\noindent
\textit{Notation.} For any tree $\ca{T} = \tup{T,E}$, we denote by
$\preceq_{\ca{T}}$ the natural ancestor relation induced by $\ca{T}$,
i.e., for $v, w \in T$, $v \preceq_{\ca{T}} w$ iff $v$ is an ancestor of $w$.

\begin{lemma}
  \label{lem:rectwc}
  If $\db{D}$ has tree-width $w$, then there is a tree decomposition
  $\delta$ of $\db{D}$ of width $w$ and an adornment $\eta$ for
  $\tup{\db{D},\delta}$ such that $\delta$ is $\eta$-simple.
\end{lemma}
\begin{proof}
  Let $\delta = \tup{\ca{T},\tup{X_v}_{v \in T}}$ be a tree
  decomposition of $\db{D}$ of width $w$. Let $\delta'$ be a tree
  decomposition of $\db{D}$ defined as follows. Initially, we define
  that $\delta'$ equals $\delta$. In a second step, we add additional
  nodes to $\delta'$. For any $v \in T$, if $|\db{D}[X_v]| = n$ then we
  add to $n - 1$ copies of $v$ to $\delta'$ that become children of $v$
  in $\delta'$. Let $v \in T$ and suppose
  $\db{D}[X_v] = \set{\alpha_1,\ldots,\alpha_n}$.  Let
  $v_1,\ldots,v_{n-1}$ be the copies of $v$ in $\delta'$. We can set
  $\eta(v) = \set{\alpha_1}$, and
  $\eta(v_i) = \set{\alpha_{i+1}}$ for $i = 1,\ldots,n-1$. It
  is easy to check that $\eta$ is an adornment for
  $\tup{\db{D},\delta'}$.

  Now $\eta$ satisfies that $|\eta(v)| = 1$, for all $v \in T'$. We can
  easily modify $\delta'$ so that $\delta'$ becomes a
  $\eta$-simple. We simply do so by successively removing from
  $\delta'$ all nodes $w \in T'$ such that there is some $v \neq w$ such
  that $\eta(v) = \eta(w)$.

  Now we show how we can modify $\delta'$ in order to become
  $\eta$-well-colored. Let $B \subseteq T'$ denote the set of black
  nodes of $\ca{T}$. Let $T^\ast$ be the smallest set such that
  \begin{enumerate*}[label={(\arabic*)}]
    \item $B \subseteq T^\ast$ and
    \item if $v$ is the greatest common ancestor of some
      $T_0 \subseteq T^\ast$, then also $v \in T^\ast$.
  \end{enumerate*}
  Hence, $T^\ast$ is $B$ closed off under greatest common ancestors. Let
  $\delta^\ast = \tup{\ca{T}^\ast, \tup{Y_v}_{v \in T^\ast}}$,
  where, for $v,w \in T^\ast$, $v \preceq_{\ca{T}^\ast} w$ iff
  $v \prec_{\ca{T}'} w$. Notice that $\delta^\ast$ is a tree
  decomposition of $\db{D}$ that has width $w$, since it contains all
  black nodes of $T'$ w.r.t.~$\eta$. It is now easy to check that
  $\delta^\ast$ is $\eta$-well-colored.
\end{proof}

\begin{lemma}
  \label{lem:factsblacknodes}
  Suppose $\eta$ is an adornment for $\tup{\db{D},\delta}$ and that
  $\delta$ is $\eta$-simple. Then $\db{D}$ contains at least as many
  facts as $\delta$ contains black nodes w.r.t.~$\eta$.
\end{lemma}
\begin{proof}
  By induction on the number $n$ of black nodes of
  $\delta = \tup{\ca{T},\tup{X_v}_{v \in T}}$ w.r.t.~$\eta$. If $n = 1$
  the claim is trivial. Suppose $\delta$ has $n + 1$ black nodes
  w.r.t.~$\eta$. There is a black node $v \in T$ such that $v$ has no
  descendant that is also black. Let
  $\delta' = \tup{\ca{T}',\tup{X_v}_{v \in T'}}$ be the tree
  decomposition that arises from $\delta$ by removing the subtree rooted
  at $v$. Let $\db{D}' = \bigcup_{v \in T'} \eta(v)$. Then
  $\delta'$ is a tree decomposition of $\db{D}'$ and $\delta'$ has
  $n$ black nodes w.r.t.~$\eta$.  Now if $\db{D}' = \db{D}$ this means
  that $\eta(v) = \eta(w)$ for some $w \neq v$. Hence, $\delta$ cannot
  be simple. Therefore, $\db{D}' \subset \db{D}$. By the induction
  hypothesis, $|\db{D}'| \geq n$ and we thus obtain $|\db{D}| \geq n + 1$.
\end{proof}

\begin{lemma}
  \label{lem:manyblacknodes}
  If $\delta$ is an $\eta$-well-colored tree decomposition of $\db{D}$,
  then the number of white nodes of $\delta$ w.r.t.~$\eta$ is strictly
  less than the number of black nodes of $\delta$ w.r.t.~$\eta$.
\end{lemma}
\begin{proof}
  Let $b_\delta$ ($w_\delta$, respectively) denote the number of black
  (white, respectively) nodes of
  $\delta = \tup{\ca{T},\tup{X_v}_{v \in T}}$ w.r.t.~$\eta$. We proceed
  by induction on the depth of $\ca{T}$, i.e., the maximum length of a
  branch leading from the root node to a leaf node. If $\ca{T}$ consists
  only of a single node and if $\delta$ is $\eta$-well-colored, this
  single node must be a black node, and so the claim holds
  trivially. (Recall that we can restrict ourselves to non-empty
    databases, since we assume non-empty schemas.) Assume that $\ca{T}$
  is of depth $n + 1$ and assume that $\delta$ is
  $\eta$-well-colored. Let $\ca{T}_1,\ldots,\ca{T}_k$ enumerate the
  subtrees of $\ca{T}$ rooted at the child nodes of the root of
  $\ca{T}$, and let $\delta_i$ ($i = 1,\ldots,k$) be the tree
  decomposition that arises from $\delta$ if we restrict $\ca{T}$ to
  $\ca{T}_i$. If the root of $\ca{T}$ is black, the claim is again
  trivial. Otherwise, if it is white, we see that $k \geq 2$ since
  $\delta$ is $\eta$-well-colored. For $i = 1,\ldots,k$, let
  $b_{\delta_i}$ ($w_{\delta_i}$, respectively) denote the number of
  black (white, respectively) nodes of $\ca{T}_i$ w.r.t.~$\eta$.  Using
  the induction hypothesis, we conclude that
  $w_{\delta} = w_{\delta_1} + \cdots + w_{\delta_k} + 1 < b_{\delta_1}
  + \cdots + b_{\delta_k} = b_{\delta}$.
\end{proof}

\noindent
\textit{Proof of~\Cref{pro:semanticrefined}.}  Assume that condition
2 does not hold. That is, there are infinitely many non-isomorphic
triples $\tup{\db{D},\delta, \eta}$ that satisfy conditions
(a)--(c). Let $S$ be the set of all these triples and let $S'$ be $S$
factorized modulo our notion of isomorphism, i.e., $S'$ contains a
representative for every isomorphism type of $S$. Let $\Phi$ be the set
$\set{\ca{T} \mid \tup{\db{D}, \delta = \tup{\ca{T},\tup{X_v}_{v
        \in T}}, \eta} \in S'}$
of trees factorized modulo usual tree isomorphism. Notice that $\Phi$
must be infinite as well. Hence, $\Phi$ must contain trees of arbitrary
size. Thus, by~\Cref{lem:manyblacknodes}, for every $k \geq 0$, we can
find a $\tup{\db{D}_k,\delta_k, \eta_k} \in S'$ such that
$\delta_k = \tup{\ca{T}_k, \tup{X_v}_{v \in T_k}}$ and
$\ca{T}_k$ has at least $k$ black nodes w.r.t.~$\eta_k$ and thus,
by~\Cref{lem:factsblacknodes}, $\db{D}_k$ has at least $k$ facts. Now
$\db{D}_k \models Q$ by assumption, but $\db{D}_0 \not\models Q$ for
every $\db{D}_0 \subset \db{D}_k$. Thus, for every $k$, we can find a
database $\db{D}$ of tree-width at most $w$ (namely $\db{D}_k$) such
that $\db{D} \models Q$, but for every $\db{D}_0 \subseteq \db{D}$ of at
most $k$ atoms we have $\db{D}_0 \not\models Q$. Hence, condition
1 does not hold.

Suppose now that condition 1 does not hold. That is, for every
$k \geq 0$ there is a database $\db{D}_k$ of tree-width at most $w$ such
that $\db{D} \models Q$, yet for every $\db{D}_0 \subset \db{D}$ of at
most $k$ facts we have $\db{D}_0 \not\models Q$. Let $S$ be the set of
all $\sche{S}$-databases $\db{D}$ of tree-width at most $w$ such that
$\db{D} \models Q$, yet for any $\db{D}_0 \subset \db{D}$, we have
$\db{D}_0 \not\models Q$. Let $S'$ be $S$ factorized modulo database
isomorphism. $S'$ must be, by assumption, infinite as
well. By~\Cref{lem:rectwc}, for every $\db{D} \in S'$, there is a tree
decomposition $\delta_{\db{D}}$ and an adornment $\eta_{\db{D}}$ of
$\tup{\db{D},\delta_{\db{D}}}$ such that $\delta_{\db{D}}$ is
$\eta_{\db{D}}$-well-colored and $\eta_{\db{D}}$-simple. Now for two
distinct $\db{D},\db{D}' \in S'$ it must be the case that
$\tup{\db{D},\delta_{\db{D}}, \eta_{\db{D}}}$ and
$\tup{\db{D}',\delta_{\db{D}'}, \eta_{\db{D}'}}$ are non-isomorphic, for
otherwise $\db{D}$ and $\db{D}'$ would be isomorphic as well. For
$\db{D} \in S'$ we know that, by construction of $S'$,
$\db{D} \setminus \set{\alpha} \not\models Q$ for all
$\alpha \in \db{D}$. Hence, the class
$\set{\tup{\db{D},\delta_{\db{D}}, \eta_{\db{D}'}} \mid \db{D} \in S'}$
is a class of infinitely many, pairwise non-isomorphic triples such that
properties (a)--(c) of condition 2 are satisfied. Thus, condition
2 does not hold as well.  \hfill$\square$

\subsection{Preliminaries: Tree Encodings}

One can naturally encode instances of bounded tree-width into trees over
a finite alphabet. Our goal here is to appropriately encode databases of
bounded tree-width in order to make them accessible to tree automata
techniques. Similar encoding techniques are well-known in the context of
guarded logics, see e.g.~\cite{BeBB16,ABeBB16} for similar encodings.

\medskip
\noindent
\underline{\smash{\textit{Labeled trees.}}} Let $\Gamma$ be an alphabet
and $(\mbb{N} \setminus \set{0})^\ast$ be the set of finite sequences of
positive integers, including the empty sequence
$\varepsilon$.\footnote{We specify that $0$ is included in $\mbb{N}$ as
  well.} Let us recall that a \emph{$\Gamma$-labeled tree} is a partial
function $t \colon (\mathbb{N} \setminus \set{0}) \rightarrow \Gamma$,
where $\dom{t}$ is closed under prefixes, i.e., $x \cdot i \in \dom{t}$
implies $x \in \dom{t}$, for all $x \in (\mbb{N} \setminus \set{0})^\ast$ and
$i \in \mbb{N} \setminus \set{0}$.  The elements contained in $\dom{t}$
identify the \emph{nodes} of $t$.  For
$i \in \mbb{N} \setminus \set{0}$, nodes of the form
$x \cdot i \in \dom{t}$ are the \emph{children} of $x$. A \emph{leaf
  node} is a node without children. The number of children of a node $x$
is its \emph{branching degree}. If every node of $t$ has branching
degree at most $m$, then we say that $t$ is $m$-ary. A \emph{path} of
\emph{length} $n$ in $t$ from $x$ to $y$ is a sequence of nodes
$x = x_1,\ldots,x_n = y$ such that $x_{i+1}$ is a child of $x_i$. A
\emph{branch} of $t$ is a path that start from the root node and ends in
a leaf node. The \emph{height} of the is the maximum length of all
branches. For $x \in \dom{t}$, we set $x \cdot i \cdot -1 = x$,
for all $i \in \mbb{N}$, and $x \cdot 0 = x$. Notice that
$\varepsilon \cdot -1$ is not defined.

\medskip
\noindent
\underline{\smash{\textit{Encoding.}}} Fix a schema $\sche{S}$ and let $w \geq 1$. Let
$U_{\sche{S},w}$ be a set containing $2w$ distinct constants. The
elements from $U_{\sche{S},w}$ will be called \emph{names}. Names are
used to encode constants of an $\sche{S}$-database of tree-width at most
$w - 1$. Neighboring nodes may describe overlapping pieces of the
encoded database. In particular, if one name is used in neighboring
nodes, this means that the name at hand refers to the same
element---this is why we use $2w$ elements for bags. Let
$\mbb{K}_{\sche{S},w}$ be the finite schema capturing the following
information:
\begin{itemize}
\item For all $a \in U_{\sche{S},w}$, there is a unary relation
  $D_a \in \mbb{K}_{\sche{S},w}$.
\item For each $R \in \sche{S}$ and every $n$-tuple
  $\ve{a} \in U^n_{\sche{S},w}$, there is a unary relation
  $R_{\ve{a}} \in \mbb{K}_{\sche{S},w}$.
\end{itemize}
Let $\Gamma_{\sche{S},w} = 2^{\mbb{K}_{\sche{S},w}}$ be an
alphabet and suppose that $\db{D}$ is an $\sche{S}$-database of
tree-width at most $w - 1$.  Consider a tree decomposition
$\delta = \tup{\ca{T},\tup{X_v}_{v \in T}}$ of $\db{D}$ that has width
at most $w-1$. Moreover, consider an adornment $\eta$ of
$\tup{\db{D},\delta}$. Fix a function
$f \colon \adom{\db{D}} \rightarrow U_{\sche{S},w}$ such that different
elements that occur in neighboring bags of $\delta$ are always assigned
different names from $U_{\sche{S},w}$. Using $f$, we can encode $\db{D}$
together with $\delta$ and $\eta$ into a $\Gamma_{\sche{S},w}$-labeled
tree $t_{\db{D},\delta,\eta}$ such that each node from $\ca{T}$
corresponds to exactly one node in $t_{\db{D},\delta,\eta}$ and vice
versa. For a node $v$ from $\ca{T}$, we denote the corresponding node of
$t_{\db{D},\delta,\eta}$ by $\hat{v}$ in the following and vice versa. In this light, the
symbols from $\mbb{K}_{\sche{S},w}$ have the following intended meaning:
\begin{itemize}
\item $D_a \in t(\hat{v})$ means that $a$ is used as a name for some
  element of the bag $X_v$.
\item $R_{\ve{a}} \in t(\hat{v})$ indicates that $R$ holds in $\db{D}$
  for the elements named by $\ve{a}$ in bag $X_v$ and this fact appears
  in $\eta(v)$.
\end{itemize}

\medskip
\noindent
\underline{\smash{\textit{Decoding trees.}}} Under certain assumptions,
we can decode a $\Gamma_{\sche{S},w}$-labeled tree $t$ into an
$\sche{S}$-database whose tree-width is bounded by $w - 1$. Let
$\names{v} = \set{a \mid D_a \in t(v)}$. We say that $t$ is
\emph{consistent}, if it satisfies the following properties:
\begin{enumerate}
\item For all nodes $v$ it holds that $|\names{v}| \leq w$.
\item For all $R_{\ve{a}} \in \mbb{K}_{\sche{S},w}$ and all
  $v \in \dom{t}$ it holds that $R_{\ve{a}} \in t(v)$ implies that
  $\ve{a} \subseteq \names{v}$.
 % \item Neighboring nodes agree on their common elements, i.e., for all
 %   $R/n \in \sche{S}$, all $\ve{a} \in U_{\sche{S},w}^n$, and all
 %   neighboring nodes $v, w$, we have that $R_{\ve{a}} \in t(v)$ and
 %   $\ve{a} \subseteq \names{w}$ implies $R_{\ve{a}} \in t(w)$.
\end{enumerate}
Suppose now that $t$ is consistent. We show how we can decode $t$ into a
database $\dec{t}$ whose tree-width is at most $w -1$. Let $a$ be a name
used in $t$. We say that two nodes $v, w$ of $t$ are
\emph{$a$-equivalent} if $D_a \in t(u)$ for all nodes $u$ on the unique
shortest path between $v$ and $w$. Clearly, $a$-equivalence defines an
equivalence relation and we let
$[v]_a =\set{\tup{w,a} \mid \text{$w$ is $a$-equivalent to
    $v$}}$ and $[v]_a^\ast = \set{w \mid \tup{w,a} \in [v]_a}$.
The domain of $\dec{t}$ is the set
$\set{[v]_a \mid v \in \dom{t}, a \in \names{v}}$ and, for
$R/n \in \sche{S}$, we define
\begin{align*}
\dec{t} \models R([v_1]_{a_1},\ldots,[v_n]_{a_n}) \iff\ &\text{there is some} \\
&\text{$v \in [v_1]_{a_1}^\ast \cap \cdots \cap [v_n]_{a_n}^\ast$}\\
&\text{s.t.~$R_{a_1,\ldots,a_n} \in t(v)$}.
\end{align*}

It is not hard to show that, if $t$ is consistent, $\dec{t}$ is
well-defined and is an $\sche{S}$-database of tree-width at most
$w - 1$. We refer the reader to \cite{ourpods}~\todo{Citation of our PODS} for proofs
of similar statements.

Given a consistent $t$, we let
$\delta_t = \tup{\ca{T},\tup{X_v}_{v \in T}}$ be a tree
decomposition of $\dec{t}$, where $\ca{T}$ is the same tree in structure
as $t$ and $X_v = \names{v}$, for all $v \in T$. Moreover, we define the adornment $\eta_t$ for $\tup{\dec{t},\delta_t}$ by
\begin{align*}
 \eta_t \colon v \longmapsto \set{R([v]_{a_1},\ldots,[v]_{a_k}) \mid R_{a_1,\ldots,a_k} \in t(v)}.
\end{align*}
We say that $t$ is \emph{simple} (\emph{well-colored}, respectively)
if $\delta_t$ is $\eta_t$-simple ($\eta_t$-well-colored,
respectively). Moreover, a node $v \in \dom{t}$ is \emph{black}
(\emph{white}, respectively), if it is black (white, respectively)
w.r.t.~$\eta_t$. We call $\delta_t$ ($\eta_t$, respectively) the
\emph{standard tree decomposition} (\emph{standard adornment}, respectively) of
$t$.

\medskip
\noindent
\underline{\smash{\textit{Bounding the branching degree.}}} For our
automata constructions that follow, it will be convenient to work on
$\Gamma_{\sche{S}}$-labeled trees whose branching degree can be bounded
by the constant $m_{\sche{S}} = 2^{\width{\sche{S}}}$ so that we
can work automata that run on $m_{\sche{S}}$-ary trees. The following
statement shows that we can always assume this without loss of
generality:

\begin{lemma}
\label{lem:branchingbounded}
Suppose $\db{D}$ is an $\sche{S}$-database and $\delta$ a tree
decomposition of $\db{D}$. Then there exists a tree decomposition
$\delta'$ of $\db{D}$ such that $\delta'$ has branching degree at most
$m_{\sche{S}} = 2^{\width{\sche{S}}}$.
\end{lemma}
\begin{proof}
  Let $\delta = \tup{\ca{T},\tup{X_v}_{v \in T}}$ be a tree
  decomposition of $\db{D}$ of width at most
  $w = \max\set{0,\width{\sche{S}}-1}$. For $v \in T$, let $d_v$
  be the branching degree of $v$ in $\ca{T}$. Moreover, let
  \begin{align*}
    d_\delta = \sum_{v \in T} \set{d_v - m_{\sche{S}} \mid d_v > m_{\sche{S}}, v \in T}.
  \end{align*}

  We are going to prove the following statement by induction on
  $d_\delta$: if $\delta$ is a tree decomposition of $\db{D}$ then there
  is a tree decomposition $\delta'$ of $\db{D}$ such that every node of
  $\delta'$ has branching degree at most
  $m_{\sche{S}} = 2^{\width{\sche{S}}}$. Moreover, $\delta'$ results
  from ``reorganizing'' nodes of $\delta$ and we can view any adornment
  $\eta$ of $\tup{\db{D},\delta}$ is also an adornment of
  $\tup{\db{D},\delta'}$.

  If $d_\delta = 0$ then the claim is trivial, since $\delta$ has no
  nodes of branching degree greater than $m_{\sche{S}}$. Suppose now
  $d_\delta = n + 1$. Let $v \in T$ be a node such that
  $d_v > m_{\sche{S}}$.  Assume $v$ is chosen such that it has, among
  all nodes of branching degree greater than $m_{\sche{S}}$, maximal
  distance to the root. Let $v_1,\ldots,v_k$ enumerate all children of
  $v$. For $i = 1,\ldots,k$, let $Y_i = X_v \cap X_{v_i}$.
  Hence, $Y_i \subseteq X_v$ and since there are at most
  $m_{\sche{S}} = 2^{\width{\sche{S}}}$ subsets of $X_v$, it must be the
  case that $Y_i = Y_j$ for some $i \neq j$. Let $\delta'$ be the tree
  decomposition that arises from $\delta$ by removing subtree rooted at
  $v_j$ from $\ca{T}$, while inserting it below $v_i$ so that $v_j$
  becomes a child node of $v_i$. Notice that $v_i$ still has branching
  degree at most $m_{\sche{S}}$ by the choice of $v$. Moreover,
  $\delta'$ is still a tree decomposition of $\db{D}$, since
  connectedness is clearly ensured. Now $d_{\delta'} = n$ and an
  application of the induction hypothesis yields the claim.
\end{proof}

By \Cref{lem:branchingbounded}, we can thus always assume that the encoding
of an $\sche{S}$-database has branching degree at most
$m_{\sche{S}} = 2^{\width{\sche{S}}}$ which we will assume for the
remainder of this section.

\subsection{Preliminaries: Two-way alternating automata (2ATA)}

For a finite set of symbols $X$, let $\mathbb{B}^+(X)$ be the set of
positive Boolean formulas that can be formed using propositional
variables from $X$, i.e., formulas using $\land, \lor$ and propositional
variables from $X$.

A \emph{two-way alternating (parity) automaton} (2ATA) on (finite)
$m$-ary trees is a tuple
$\ata{A} = \tup{S, \Gamma, \delta, s_0, \Omega, \dir}$, where
\begin{itemize}
  \item $S$ is a finite set of \emph{states},
  \item $\Gamma$ is the \emph{input alphabet},
  \item
    $\delta \colon S \times \Gamma \rightarrow
    \mathbb{B}^+(\mathsf{tran}(\ata{A}))$
    is the \emph{transition function}, where
    $\mathsf{tran}(\ata{A}) = \set{\ndia{d}s, \nbox{d}s \mid d
      \in \dir})$,
  \item $s_0$ is the \emph{initial state},
  \item $\Omega \colon S \rightarrow \mathbb{N}$ is the \emph{parity
      condition} that assigns to each $s \in S$ a \emph{priority}
    $\Omega(s)$.
  \item $\dir$ is a set of \emph{directions} and, in our case, always
    equals $\set{-1,0,1,\ldots,m}$.\footnote{Notice that we always use
      the same set of directions here, i.e., we work on trees of a fixed
      branching degree. We nevertheless make this set of directions
      explicit to avoid confusion, since our cost automata are going to
      work with amorphous automata that work on trees of arbitrary
      branching degree.}
\end{itemize}
Notice that we make explicit the set of directions the automaton may
use. Formally, a direction is just a function that maps a node to other
nodes. For $d \in \set{-1,0,1,\ldots,m}$, we set
\begin{align*}
d \colon &\varepsilon \longmapsto \set{d}, & &\text{if $d \neq -1$,}\\
&v \longmapsto \set{v \cdot d},&  &\text{for $v \neq \varepsilon$.}
\end{align*}
Notice that $-1(\varepsilon)$ is thus undefined, since the root has no
parent. The direction $0$ maps every node to itself and thus indicates
that the automaton should stay in the current node, the direction $-1$
indicates the automaton should proceed to the parent node, and a
direction $k \in \set{1,\ldots,m}$ indicates that the automaton should
move to the $k$-th child of the current node. Transitions of the form
$\ndia{d}s$ mean that a copy of the automaton must accept in state $s$
for at least one node in direction $d$, while the dual connective,
$\nbox{d}s$, means that, for every neighbor in direction $d$, if a copy
of the automaton is sent to that neigbor in state $s$, it must
accept. Notice that our automaton is two-way, since it can proceed to
the in both directions---to the parent and to the children.

\begin{remark}
We will consider our 2ATA to run on finite trees. The parity
condition nevertheless makes sense, since our automata are two-way and
two-way movements can give rise to infinite runs.
\end{remark}

Given an $\ata{A}$ as above and a $\Gamma$-labeled input tree $t$, the
notion of acceptance of $t$ is defined via a game played between two
players, Eve and Adam. The goal of Eve is to satisfy the parity
condition and prove that $t$ is accepted by $\ata{A}$, while to goal of
Adam is to disprove this. We shall make this more precise in the
following.

Let $\chi \in \mathbb{B}^+(\mathsf{tran}(\ata{A}))$ be a positive
formula. We assign $\chi$ to an \emph{owner} according to its form:
\begin{itemize}
\item If $\chi = \chi_1 \land \chi_2$ (respectively,
  $\chi = \chi_1 \lor \chi_2$) then $\chi$ is owned by Adam
  (respectively, Eve).
\item If $\chi = \nbox{d}s$ (respectively, $\ndia{d}s$) then $\chi$ is
  owned by Adam (respecitvely, Eve).
\end{itemize}
The \emph{acceptance game $\ca{G}(\ata{A},t)$ for $\ata{A}$ and $t$}
is played in the arena
$\mathbb{B}^+(\mathsf{tran}(\ata{A})) \times \dom{t}$. For each position $\tup{\chi, v}$ of the arena, we define the set of \emph{possible choices}:
\begin{itemize}
\item If $\chi = \chi_1 \land \chi_2$ or $\chi = \chi_1 \lor \chi_2$
  then the possible choices are $\set{\tup{\chi_1,v}, \tup{\chi_2,v}}$.
\item If $\chi = \nbox{d}s$ or $\chi = \ndia{d}s$ then the possible choices are
    $\set{\tup{\delta(s,t(w)), w} \mid w \in d(v)}$.
\end{itemize}
Let $\chi_0 = \delta(s_0,t(\varepsilon))$. The initial position
of the game $\ca{G}(\ata{A},t)$ is
$\tup{\chi_0, \varepsilon}$ and from any position
$\tup{\chi,v}$:
\begin{itemize}
\item The player that owns $\chi$ selects a $\tup{\chi',w}$ among the
  possible choices of $\tup{\chi, v}$, and
\item the game continues from position $\tup{\chi', w}$.
\end{itemize}

The transition from $\tup{\chi,v}$ to $\tup{\chi',w}$ is called a
\emph{move}. By \emph{play} in $\ca{G}(\ata{A},t)$ we mean a sequence of
moves
$\tup{\chi_0, \varepsilon},\tup{\chi_1,v_1},\tup{\chi_2,v_2},\ldots$
(recall that $\chi_0 = \delta(s_0,t(\varepsilon))$. A \emph{strategy}
for one of the players is a function that returns the next choice for
that player given the history of the play. Fixing a strategy for both
players thus uniquely determines a play in $\ca{G}(\ata{A},t)$. A play
$\pi$ is \emph{consistent} with a strategy $\xi$ if there is a strategy
$\xi'$ for the other player such that $\xi$ and $\xi'$ yield $\pi$. 

We say that a strategy is \emph{winning for Eve}, if every play
consistent with it \emph{satisfies} the parity acceptance condition, that is,
if every play
$\tup{\chi_0, \varepsilon},\tup{\chi_1,v_1},\tup{\chi_2,v_2},\ldots$
consistent with that strategy, the maximum priority among
$\Omega(\chi_0),\Omega(\chi_1),\Omega(\chi_2),\ldots$ that occurs
infinitely often is even. Here, we set
\begin{align*}
\Omega(\chi) =
\begin{cases}
  \Omega(s), & \text{if $\chi = \ndia{d}s$ or $\chi = \nbox{d}s$,}\\
  \min \Omega(S), & \text{otherwise.}
\end{cases}
\end{align*}
The \emph{language of $\ata{A}$}, denoted
$\ca{L}(\ata{A})$, is the set of all $\Gamma$-labeled trees $t$ such that
Eve has a winning strategy in $\ca{G}(\ata{A},t)$.

\subsection{Proof of~\Cref{lem:ataconsistency}}

The construction of this automaton is fairly standard and we only make a
few comments on it (cf.~\cite{BeBB16} for a similar
construction). Notice first of all that each of the two conditions for
consistency can be checked separately, and taking the intersection of
the respective automata yields the desired automaton. Each of the two
consistency conditions involves a top-down pass through the tree, while
checking the respective condition locally.

\subsection{Proof of~\Cref{lem:atarectwellc}}

We can devise $\ata{R}_{\sche{S}}$ as the intersection of two separate
2ATA, $\ata{R}_{1,\sche{S}}$ and $\ata{R}_{2, \sche{S}}$, where the
former checks whether the input tree $t$ is simple and the latter checks
whether $t$ is well-colored. It is well-known that building the
intersection of two 2ATA is feasible in polynomial time. Recall that
$\delta_t$ denotes the standard tree decomposition of $t$ and $\eta_t$
the standard adornment of $\tup{\dec{t},\delta_t}$.

\medskip
\noindent
\underline{\smash{\textit{The automaton $\ata{R}_{1,\sche{S}}$.}}}  In order to check whether $t$ is simple, we
have to check two conditions:
\begin{enumerate*}[label={(\roman*)}]
  \item whether $|\eta_t(v)| \leq 1$ for all $v \in \dom{t}$ and
  \item whether $\eta_t(v) \neq \eta_t(w)$ for all $v \neq w$.
\end{enumerate*}
The first condition is easy to check (respecting the stated size bounds)
and we leave this as an exercise for the reader. We describe how the
check the second one and assume that the input tree satisfies the first
condition.

We shall describe the game $\ca{G}(\ata{R}_{1,\sche{S}},t)$. Adam will
have a winning strategy in $\ca{G}(\ata{R}_{1,\sche{S}}, t)$ iff
$\eta_t(v) = \eta_t(w)$ for some nodes $v \neq w$. Adam first navigates
to an arbitrary node $v$ for which he wants to prove that there is some
other $w \neq v$ such that $\eta_t(v) = \eta_t(w)$. He then selects the
one and only atom $R_{\ve{a}} \in t(v)$ and guesses the path to the node
$w$ for which he thinks that $\eta_t(v) = \eta_t(w)$. If he finds that
node, he wins. By navigating to $w$, he must remember the atom
$R_{\ve{a}}$ in the states and also the direction he came from. He must
remember the direction due to the fact that we require $v \neq w$. Due
to that, the number of states of $\ata{R}_{1,\sche{S}}$ also depends
linearly on the branching degree $m_{\sche{S}}$ which still allows us
the respect the stated size bounds since
$m_{\sche{S}} = 2^{\width{\sche{S}}}$. Now while navigating to $w$, Adam
is not allowed to traverse the tree backwards in the direction he came
from. For storing this information, we need exponentially many states in
$\width{\sche{S}}$ and linearly many in $|\sche{S}|$ and $m_{\sche{S}}$.

\medskip
\noindent
\underline{\smash{\textit{The automaton $\ata{R}_{2,\sche{S}}$.}}}
Recall that a node $v \in \dom{t}$ is well-colored if it is either black
or it has at least two successor nodes which are both well-colored.
Having this definition in place, devising $\ata{R}_{2,\sche{S}}$ becomes
quite easy. In $\ca{G}(\ata{R}_{2, \sche{S}},t)$, Adam guesses the node
$v$ with a maximum distance from the root of which he wants to prove
that this node is not good. Since $v$ is not good, $v$ is white and it
has less than two successors that are good. Moreover, since $v$ has
maximum distance from the root, $v$ must have either no successors or it
has a single successor that is black. (Two black successors would turn
$v$ to a good node, while one black and a white successor of which both
are not good would turn the white successor to a non-good node which has
higher distance to the root.) Therefore, all Adam has to do is to
challenge Eve to show the existence of the (non-existent) second
successor. Adam will win if Eve cannot point to such a second
successor. Notice that the size of the state set of this automaton is
independent from $\sche{S}$. In fact, $\ata{R}_{2,\sche{S}}$ has
constantly many states.

\subsection{Proof of~\Cref{lem:atasat}}

The construction of this automaton appears in~\cite{ourpods}. Notice
that in~\cite{ourpods}, this automaton is devised for input trees whose
decodings are \emph{acyclic}\footnote{A definition of acyclicity will be
  given in \Cref{sec:treeification}.} rather than of tree-width at most
$\max\set{0, \width{\sche{S}}-1}$. However, the construction works also
with our encodings. Alternatively, one can view $\ata{A}_{Q}$ as a
version of the cost automaton $\ata{H}_Q$ from~\Cref{lem:costata} that
has no counters at all.

\subsection{Proof of~\Cref{lem:atamin}}

We shall informally describe the construction of $\ata{M}_Q$ and
describe its size bounds. Then we are going to prove that
$\ata{M}_Q$ indeed can be used to check whether $Q$ is FO-rewritable.

Firstly, we define an auxiliary alphabet $\Lambda_{\sche{S}}$ that is a
copy of (some parts of) the alphabet $\Gamma_{\sche{S}}$. More
specifically, for every $\rho \in \Gamma_{\sche{S}}$ such that
\begin{align*}
\rho \cap \set{R_{\ve{a}} \mid R/n \in \sche{S}, \ve{a} \in
  U_{\sche{S},\width{\sche{S}}}^n} = \set{\alpha_1,\ldots,\alpha_k},
\end{align*}
we stipulate that $\Lambda_{\sche{S}}$ contains the symbol
$\rho^\discnode = \set{\alpha^\discnode_1,\ldots,\alpha^\discnode_k}$.
That is, the alphabet $\Lambda_{\sche{S}}$ carries information on the
facts named in a $\Gamma_{\sche{S}}$-labeled tree (we call facts of the
form $\alpha^\discnode$ \emph{tagged}). Intuitively, a fact of the form
$R^\discnode_{\ve{a}}$ specifies that the minimization procedure (that
is yet to be implemented in $\ca{M}_Q$) should aim to satisfy $Q$ without
the need of $R_{\ve{a}}$.

For a $\Lambda_{\sche{S}}$-labeled tree $t$, we define
$t \upharpoonright \Gamma_{\sche{S}}$ as the $\Gamma_{\sche{S}}$-labeled
$t'$ that arises from $t$ by setting
$t'(v) = t(v) \cap \bigcup \Gamma_{\sche{S}}$ for all
$v \in \dom{t'}$. We say that $t$ is an \emph{extension} of
$t' = t \upharpoonright \Gamma_{\sche{S}}$. We call $t$
\emph{consistent} if $t \upharpoonright \Gamma_{\sche{S}}$ is consistent
and at least one node $v \in \dom{t}$ is labeled with a fact of the form
$R_{\ve{a}}^\discnode$ and, moreover, for all $v \in \dom{t}$, if
$R^\discnode_{\ve{a}} \in t(v)$ then also $R_{\ve{a}} \in t(v)$. We
define $\dec{t}^-$ to be
\begin{align*}
  \dec{t \upharpoonright \Gamma_{\sche{S}}} \setminus \set{R([v]_{a_1},\ldots,[v]_{a_k}) \mid R^\discnode_{a_1,\ldots,a_k} \in t(v)}.
\end{align*}
That is, in $\dec{t}^-$ we remove the facts that are tagged.

\begin{lemma}
  \label{lem:removeata}
  There is a 2ATA $\ata{D}_Q$ that runs on $m$-ary
  $\Lambda_{\sche{S}}$-labeled trees and accepts a
  $\Lambda_{\sche{S}}$-labeled tree if and only if:
\begin{enumerate}
\item $t$ is consistent,
\item $\dec{t}^- \models Q$, i.e., $t$ without the tagged facts satisfies
  $Q$.
\end{enumerate}
The number of states of $\ata{D}_Q$ is exponential in $\width{\sche{S}}$
and linear in $|\sche{S} \cup \sig{\ont{O}}|$. Moreover, $\ata{D}_Q$ can
be constructed in double exponential time in the size of $Q$.
\end{lemma}
\noindent
$\ata{D}_Q$ can be constructed as the
intersection of the 2ATA, where one checks consistency and the other
ensures that $\dec{t}^- \models Q$.  The former can be constructed in a
similar spirit as $\ata{A}_Q$ from~\Cref{lem:atasat} so that the
construction of $\ata{D}_Q$ respects the same size bounds. Moreover,
consistency of $t$ can be checked in a similar fashion as consistency
for $\Gamma_{\sche{S}}$-labeled trees, with an additional check that the
input tree contains at least one tagged fact.

Having $\ata{D}_Q$ from \Cref{lem:removeata} in place, we are now going
to construct $\ata{M}_Q$. $\ata{M}_Q$ will accept a consistent
$\Gamma_{\sche{S}}$-labeled input tree $t$ if and only if there is no
$\Lambda_{\sche{S}}$-labeled extension $t'$ of $t$ such that
\begin{enumerate*}[label={(\roman*)}]
  \item $t'$ is consistent and
  \item $\dec{t'}^- \models Q$.
\end{enumerate*}
Equivalently, $\ata{M}_Q$ will accept a consistent $t$ iff there are no
facts $\alpha_1,\ldots,\alpha_k \in \dec{t}$ ($k \geq 1$) such that
$\dec{t} \setminus \set{\alpha_1,\ldots,\alpha_k} \models Q$.

We can, according to~\cite{Va98}, convert $\ata{D}_Q$ into a
\emph{nondeterministic parity tree automaton} on $m_{\sche{S}}$-ary
trees $\ata{D}'_Q$ which is simply a 2ATA on $m_{\sche{S}}$-ary trees
where all transitions are of the form
$\delta(q, a) = \bigvee_{i\in I} (\ndia{1}q_{1,i} \land \cdots \land
\ndia{m_{\sche{S}}}q_{m,i})$.
This conversion causes an exponential blowup on the size of the state
set. We can view $\Lambda$ as an alphabet extending $\Gamma$. Hence, we
can perform the operation of \emph{projection} on $\ata{D}'_Q$ in such a
way that the resulting automaton, call it $\exists\ata{D}_Q$, accepts
$\Gamma_{\sche{S}}$-labeled trees only. Notice that projection is easy
to perform in the case of nondeterministic parity automata. Indeed, in
order to construct $\exists \ata{D}_Q$, the only thing we have to do is
to guess symbols from $\Lambda_{\sche{S}}$ in the transition function of
$\exists \ata{D}_Q$. This does not involve a blowup on the state set,
and $\exists\ata{D}_Q$ can be constructed in polynomial time in the size
of $\ata{D}'_Q$.
Notice that $\exists\ata{D}_Q$ accepts a consistent
$\Gamma_{\sche{S}}$-labeled tree $t$ iff there is a
$\Lambda_{\sche{S}}$-labeled extension $t'$ of $t$ such that $t'$ is
consistent and $\dec{t'}^- \models Q$. We thus obtain $\ata{M}_Q$ from
$\exists \ata{D}_Q$ by building the complement of $\exists \ata{D}_Q$.
Building the complement of a 2ATA is easy---we simply swap the
formulas owned by Adam and Eve.

\subsection{Proof of~\Cref{theorem:atafin}}

Let $Q = \tup{\sche{S},\ont{O},G}$. As said in the main body of the
paper, we can obtain $\ata{B}_Q$ by intersecting the respective 2ATA
from \Cref{lem:ataconsistency,lem:atarectwellc,lem:atasat,lem:atamin}.

It is clear $\ata{B}_Q$ has double-exponentially many states in
$\width{\sche{S}}$ and, moreover, $\ata{B}_Q$ can be constructed in
double-exponential time. Thus, $\ata{B}_Q$ accepts a
$\Gamma_{\sche{S}}$-labeled tree $t$ if and only if
\begin{itemize}
  \item $t$ is consistent,
  \item $t$ is well-colored and simple,
  \item $\dec{t} \models Q$, and
  \item $\dec{t} \setminus \set{\alpha} \not\models Q$ for all
    $\alpha \in \dec{t}$.
\end{itemize}

For a proof of
\Cref{theorem:atafin} it thus remains to be shown that the language of $\ata{B}_Q$ is infinite iff $Q$ is not
FO-rewritable.

Suppose first that $\ca{L}(\ata{B}_Q)$ is infinite. Since the the
branching degree of the input trees is bounded (recall that we run on
$m_{\sche{S}}$-ary trees), $\ata{B}_Q$ accepts trees of arbitrary
height. Hence, there are infinitely many trees
$t_{0},t_{1},\ldots,t_{k},\ldots$ and natural numbers
$h_0,h_1,\ldots,h_k,\ldots$ such that, for $i \geq 0$,
\begin{itemize}
\item $t_i$ has height $h_i$,
\item $\dec{t_i} \models Q$,
\item $t_i$ is well-colored and simple, and
\item $\dec{t_i} \models Q$, while
  $\dec{t_i} \setminus \set{\alpha} \not\models Q$ for all
  $\alpha \in \dec{t_i}$.
\end{itemize}
Moreover, we can assume that $i \neq j$ implies $h_i \neq h_j$
(otherwise we simply drop $t_{j}$). For $i \geq 0$, consider the
standard tree decomposition $\delta_{t_i}$ and the standard adornment
$\eta_{t_i}$ of $t_{i}$. It is clear that $\delta_{t_i}$ is
$\eta_{t_i}$-well-colored and $\eta_{t_i}$-simple as well. Moreover, for
$i \neq j$, the triples $\tup{\dec{t_{i}}, \delta_{t_i}, \eta_{t_i}}$ must
be non-isomorphic, since the height of $t_{i}$ and $t_{j}$ differ,
i.e., $h_i \neq h_j$. We thus obtain by~\Cref{pro:semanticrefined} that
$Q$ cannot be FO-rewritable.

Suppose now that $Q$ is not FO-rewritable. By~\Cref{pro:semanticrefined}
there is an infinite class $\ca{S}$ of pairwise non-isomorphic triples
$\tup{\db{D},\delta, \eta}$ (where $\db{D}$ is an $\sche{S}$-database,
$\delta$ a tree decomposition of $\db{D}$ of width at most
$\max\set{0,\width{\sche{S}}-1}$, and $\eta$ an adornment of
$\tup{\db{D},\delta}$) such that
\begin{itemize}
\item $\delta$ is $\eta$-well-colored and $\eta$-simple,
\item $\db{D} \models Q$, and
\item for every $\alpha \in \db{D}$ it holds that
  $\db{D}\setminus \set{\alpha} \not\models Q$.
\end{itemize}

Recall that we can encode every such triple
$\gamma = \tup{\db{D},\delta, \eta}$ as a $\Gamma_{\sche{S}}$-labeled
tree $t_\gamma$. Considering the encoding, for two non-isomorphic
triples $\gamma,\gamma' \in \ca{S}$ we must have
$t_\gamma \neq t_{\gamma'}$.  By construction we then have
$\set{t_\gamma \mid \gamma \in \ca{S}} \subseteq \ca{L}(\ata{B}_Q)$.
Hence, $\ca{L}(\ata{B}_Q)$ must be infinite since $\ca{S}$ is. This
completes the proof of \Cref{theorem:atafin}.

\subsection{Proof of \Cref{coro:ata-approach}}

As described in the main body of this paper, in order to decide
$\forew{\class{G}}{\class{AQ}_0}$, it suffices to decide whether the
language of the 2ATA $\ata{B}_Q$ from \Cref{theorem:atafin} is
finite. To this end, we first convert $\ata{B}_Q$ into a
non-deterministic parity tree automaton $\ata{B}_Q'$ according to the
procedure presented in~\cite{Va98}.  The number of states of
$\ata{B}_Q'$ is exponential in the number of states of
$\ata{B}_Q$. Since $\ata{B}_Q'$ accepts only finite trees, we can view
$\ata{B}_Q'$ as a conventional bottom-up tree automaton that works on
finite trees. We can then check whether $\ca{L}(\ata{B}_Q')$ is finite
in polynomial time in the size of $\ata{B}_Q'$~\cite{Va92}. Since
$\ata{B}_Q'$ has triple exponentially many states, the \threeexp\ upper
bound of $\forew{\class{G}}{\class{AQ}_0}$ follows. In case of bounded
arity, $\ata{B}_Q'$ has double exponentially many states, which yields
the \twoexp\ upper bound as stated.

\section{Proofs for \Cref{sec:secondapproach}}

\subsection{Proof of \Cref{pro:semanticcost}}

Let $Q = \tup{\sche{S},\ont{O},G}$ be an OMQ from
$\tup{\class{G},\class{AQ}_0}$.

\begin{lemma}
  \label{lem:derivtree}
  Let $\db{D}$ be an $\sche{S}$-database. Then $\db{D} \models Q$ iff
  there is a derivation tree for $\db{D}$ and $Q$.
\end{lemma}
\begin{proof}[sketch]
  The direction from right to left is straightforward and left to the
  reader. For the other direction, we remark that a proof of a similar
  statement appears in~\cite{ourpods}. We only sketch the
  idea. Basically, since $\ont{O}$ consists of guarded rules only, for
  any fact $R(\ve{a})$ such that
  $\tup{\db{D},\ont{O}} \models R(\ve{a})$, one can find a
  \emph{guarded} sequence of facts $\alpha_1,\ldots,\alpha_k$ such that
  $\tup{\set{\alpha_1,\ldots,\alpha_k},\ont{O}} \models R(\ve{a})$. (We
  say that $\alpha_1,\ldots,\alpha_k$ is \emph{guarded}, if there is an
  $i = 1,\ldots,k$ such that
  $\adom{\set{\alpha_i}} \supseteq
  \adom{\set{\alpha_1,\ldots,\alpha_k}}$.)
  One then builds an appropriate derivation tree $\ca{T}$ by starting
  with the root node, labeled by $G$, and successively searching for
  such guarded sequences of facts. The facts contained in the guarded
  sequence then become labels of child nodes of the current
  node. Continuing this process recursively then gives rise to a
  derivation tree for $\db{D}$ and $Q$. (It is pretty easy to check
  that, if $\tup{\db{D},\ont{O}} \models R(\ve{a})$, the labels of leaf
  nodes of $\ca{T}$ must be contained in $\db{D}$.)

  It remains therefore to be argued why the branching degree of $\ca{T}$
  can be bounded by $k_Q = |\sche{S} \cup \sig{\ont{O}}|\cdot w^w$,
  where $w = \width{\sche{S} \cup \sig{\ont{O}}}$. This is
  simply the case, since any guarded sequence $\alpha_1,\ldots,\alpha_k$
  with more than $k_Q$ facts must contain repetitions: for a fixed guard
  $\alpha_i$ ($i = 1,\ldots,k$), there are at most $w^w$ different
  sequences of constants that use elements from
  $\adom{\set{\alpha_i}}$. Moreover, there are
  $|\sche{S} \cup \sig{\ont{O}}|$ different symbols that we can attach
  to such sequences. Therefore, $\alpha_1,\ldots,\alpha_k$ contains at
  most $k_Q$ distinct facts.
\end{proof}

\begin{remark}
  Notice that from this proof it becomes evident that we can restrict
  ourselves to derivation trees where the set of children of a
  non-leaf node is guarded in the above sense. We will assume this of
  all derivaton trees in the following.
\end{remark}

\noindent
\textit{Proof of \Cref{pro:semanticcost}.}  Throughout the proof, we let
$w = \max\set{0, \width{\sche{S}} -1}$.

Suppose first that there is a $k \geq 0$ such that, for every
$\sche{S}$-database $\db{D}$ of tree-width at most $w$, if
$\db{D} \models Q$, then there is a $\db{D}' \subseteq \db{D}$ with at
most $k$ facts such that $\db{D}' \models Q$. We show that $\costq{Q}$
is finite. Let $\ca{S}$ be the set of all $\sche{S}$-databases $\db{D}$
of at most $k$ facts such that $\db{D} \models Q$, and consider $\ca{S}$
to be factorized modulo database isomorphism. Clearly, $\ca{S}$ must be
finite. For each $\db{D} \in \ca{S}$, let $\ca{T}_{\db{D}}$ be a
derivation tree for $\db{D}$ and $Q$. Let
$n = \max\set{\height{\ca{T}_{\db{D}}} \mid \db{D} \in \ca{S}}$.
We claim that $\costq{Q} \leq n$. Indeed, suppose that $\db{D}$ is an
$\sche{S}$-database of tree-width at most $w$ such that
$\db{D} \models Q$. By assumption, we can find a
$\db{D}' \subseteq \db{D}$ of at most $k$ atoms such that
$\db{D}' \models Q$. Hence, (some isomorphic copy of) $\db{D}'$ is
contained in $\ca{S}$. Consider an arbitrary derivation tree $\ca{T}$
for $\db{D}$ and $Q$. In case $\height{\ca{T}} > n$, we know that
$\ca{T}_{\db{D}'}$ is also a derivation tree for $\db{D}$ and $Q$, and
so $\ca{T}$ is not the minimal one. Therefore, $\height{\ca{T}} \leq n$
and so $\costq{Q} \leq n$, as required.

Suppose now that $\costq{Q} = n$ for some $n \in \mathbb{N}$. For a
derivation tree $\ca{T}$, let $\ell_{\ca{T}}$ denote the number of leaf
nodes of $\ca{T}$. Let
\begin{align*}
  k = \sup\set{\ell_{\ca{T}} \mid\ &\text{$\ca{T}$ is a derivation tree for} \\ 
 &\text{$\db{D}$ and $Q$ of minimum height, where}\\
 &\text{$\db{D}$ is an $\sche{S}$-database with $\tw{\db{D}} \leq w$}}.
\end{align*}
Notice that $k$ exists (i.e., $k \in \mathbb{N}$), since we can bound
the height of the derivation trees used in the definition of $k$ by $n$,
and since the number of leaf nodes of a derivation tree of finite height
cannot be arbitrarily large (recall that the branching degree of a
derivation tree is bounded). We claim that $k$ is the bound we are
looking for in condition 1 of \Cref{pro:semanticcost}. Let $\db{D}$ be
an $\sche{S}$-database of tree-width at most $w$ such that
$\db{D} \models Q$. Consider a derivation tree $\ca{T}$ for $\db{D}$ and
$Q$ of minimum height. Let $\alpha_1,\ldots,\alpha_m$ be the leaf nodes
of $\ca{T}$. By construction, the number of leaf nodes of $\ca{T}$ is
surely bounded by $k$, i.e., $m \leq k$. Moreover, $\ca{T}$ is by
definition also a derivation tree for
$\db{D}' = \set{\alpha_1,\ldots,\alpha_m}$ and $Q$. By
\Cref{lem:derivtree} we then have $\db{D}' \models Q$. Now $\db{D}'$ is
a subset of $\db{D}$ of at most $k$ atoms. Therefore, condition 1 of the
statement of \Cref{pro:semanticcost} holds as well.  \hfill$\square$

\subsection{Preliminaries: Cost automata}

We remark that in this section we are going to work on labeled trees
that are amorphous, i.e., that have an arbitrary branching degree. This
is not necessary for a technical reason, but simplifies presentation of
the cost automaton from~\Cref{lem:costata}.

\medskip
\noindent
\underline{\smash{\textit{Objectives.}}} An \emph{objective} is a triple
$\obj = \tup{\act, f, \goal}$, where $\act$ is a finite set of
\emph{actions}, $f$ a function (the \emph{objective function}) that
assigns values from $\mathbb{N}_\infty$ to sequences of actions, and
$\goal \in \set{\min,\max}$. We shall consider a run of a (two-way)
alternating cost automaton as a two-player game with players Eve and
Adam, where $\goal$ specifies whether Eve's aim is to minimize or
maximize the objective function.

An example for an objective can be given by the well-known parity
acceptance condition which we also used for plain 2ATA. This condition
can be recast into a \emph{parity objective}
$\parity = \tup{P, \mathsf{cost}_{\parity},\goal}$, where $P$ is a
finite set of \emph{priorities} and $\mathsf{cost}_{\parity}$ is
specified as follows: if $\goal = \min$ ($\goal = \max$, respectively)
then $\mathsf{cost}_{\parity}$ maps a sequence of priorities to $0$
($\infty$, respectively) if the maximum priority that occurs infinitely
often is even, and to $\infty$ ($0$, respectively) otherwise.

\medskip
\noindent
\underline{\smash{\textit{Cost automata model.}}} Let $\Gamma$ be an
alphabet. A \emph{two-way alternating cost automaton} $\ata{A}$ on
$\Gamma$-labeled trees is a tuple
$\tup{S, \Gamma, s_0, \dir, \obj, \delta}$, where
\begin{itemize}
  \item $S$ is a finite set of \emph{states}.
  \item $s_0$ is the \emph{initial state}.
  \item $\dir$, as in the case of 2ATA, describes the possible
    directions; in our case, we always have
    $\dir = \set{0, \updownarrow}$, where
  \begin{align*}
    \updownarrow \colon &\varepsilon \longmapsto \set{i \mid i \in \set{1,\ldots,n_\varepsilon}},\\
                        &v \longmapsto \set{v \cdot i \mid i \in \set{-1,1,\ldots,n_v}},\ \text{for $v \neq \varepsilon$},
  \end{align*}
  and $n_v$ denotes the number of successors of $v$. Hence, the
  direction $\updownarrow$ denotes all possible neighbors of a node,
  including the parent (the root $\varepsilon$ has no parent and
  therefore ${\updownarrow}(\varepsilon)$ only includes all children of
  $\varepsilon$).
  \item $\obj$ is an objective;
  \item
    $\delta \colon S \times \Gamma \rightarrow
    \mathbb{B}^+(\mathsf{tran}(\ata{A}))$ the \emph{transition function}, where
  \begin{align*}
    \mathsf{tran}(\ata{A}) = \set{\ndia{d}\tup{s, c}, \nbox{d}\tup{s,c} \mid\ &s \in S, c \in \act,\\ &d \in \dir}.
  \end{align*} 
\end{itemize}
To emphasize the objective that is used, we often call an automaton in
the form of $\ata{A}$ an \emph{$\obj$-automaton}.

Notice that each transition also carries information on the action that
is to be performed when switching to a new state. We will present the
concrete actions available to our automata model below. Also note that
our cost automata work on trees of arbitrary branching
degree. 

As in the case of 2ATA, we assign owners to each formula from
$\mathbb{B}^+(\mathsf{tran}(\ca{A}))$ in the expected manner. That is,
conjunctions are owned by Adam, disjunctions are owned by Eve, atomic
formulas of the form $\nbox{d}\tup{s,c}$ are owned by Adam, while those
of the form $\ndia{d}\tup{s,c}$ are owned by Eve. 

Let $t$ be a $\Gamma$-labeled tree. As in the case of 2ATA, we define a
two-player \emph{(cost) acceptance game $\ca{G}(\ata{A},t)$ for $\ata{A}$ and
  $t$}. The arena of the game is again
$\mathbb{B}^+(\mathsf{tran}(\ca{A})) \times \dom{t}$, and the notion of
\emph{possible choices} of a position $\tup{\chi,v}$ in the game is
defined as in the case of 2ATA:
\begin{itemize}
\item If $\chi = \chi_1 \land \chi_2$ or $\chi = \chi_1 \lor \chi_2$
  then the possible choices are $\set{\tup{\chi_1,v}, \tup{\chi_2,v}}$.
\item If $\chi = \nbox{d}\tup{s,a}$ or $\chi = \ndia{d}\tup{s,a}$ then
  the possible choices are
  $\set{\tup{\delta(s,t(w)), w} \mid w \in d(v)}$.
\end{itemize}
Let $\chi_0 = \delta(s_0, t(\varepsilon))$. The
initial position of the game is $\tup{\chi_0, \varepsilon}$ and from any
position $\tup{\chi, v}$:
\begin{itemize}
\item The player that owns $\chi$ selects a $\tup{\chi',w}$ from the
    possible choices of $\tup{\chi,v}$, and
  \item the game continues from position $\tup{\chi', w}$.
\end{itemize}

The transition from $\tup{\chi, v}$ to $\tup{\chi', w}$ is a
\emph{move}. If $\chi$ is of the form $\ndia{d}\tup{s,c}$ or
$\nbox{d}\tup{s,c}$, then we say the \emph{output} of that move is
$c$. Otherwise, we simply say that that move has no output. A
\emph{play} in $\ca{G}(\ata{A},t)$ is a sequence of moves
$\tup{\chi_0, \varepsilon}, \tup{\chi_1,v_1},\tup{\chi_2,v_2},\ldots$
and a \emph{strategy} for one of the players is a function that, given
the history of the play, returns the next choice for that player. Again,
fixing a strategy for both players uniquely determines a play in
$\ca{G}(\ata{A},t)$. A play $\pi$ is \emph{consistent} with a strategy
$\xi$ if there is a strategy $\xi'$ for the other player such that $\xi$
and $\xi'$ yield $\pi$. The \emph{output} of a play
$\pi = \tup{\chi_0, v_0}, \tup{\chi_1,v_1},\tup{\chi_2,v_2},\ldots$ is
the sequence of actions $c_{i_0},c_{i_1},c_{i_2},\ldots$ such that, for
all $j\geq0$, $c_{i_j}$ is the output of $\tup{\chi_j,v_j}$.

Suppose $\obj = \tup{\act,f,\goal}$. The \emph{cost} of a play $\pi$
consistent with a winning strategy of Eve is the value of $f$ on the
output of that play. If $\goal = \min$ ($\goal = \max$, respectively)
then an \emph{$n$-winning strategy} for Eve is a strategy such that the
cost of any play (also called the \emph{cost} of that play) consistent
with that strategy is at most $n$ (at least $n$, respectively). We
define
\begin{align*}
  \val{\ata{A}}(t) = \mathrm{op}\set{n \mid \text{Eve has an $n$-winning strategy in $\ca{G}(\ata{A},t)$}},
\end{align*}
where $\mathrm{op} = \inf$ ($\mathrm{op} = \sup$, respectively) if
$\goal = \min$ ($\goal = \max$, respectively). We say that $\val{\ata{A}}$ is
the \emph{cost function defined by $\ata{A}$}.

\medskip
\noindent
\underline{\smash{\textit{Counter actions.}}} As in~\cite{BeCCB15}, we
are interested in objectives that are based on \emph{counters}. We use
the elementary actions \emph{increment \& check} $\ac{ic}$, \emph{reset}
$\ac{r}$, and \emph{no change} $\varepsilon$. Let $\gamma$ be a
counter. Its initial value is $0$ and afterwards it can take values from
$\mathbb{N}$ according to a sequence $\ve{u}$ of actions from
$\set{\ac{ic},\ac{r},\varepsilon}$. The meanings of $\ac{r}$ and
$\varepsilon$ are clear. The operation $\ac{ic}$ increments the counter
value (the \emph{increment}) and, at the same time, indicates that we
are interested in the current value of the counter (the
\emph{check}). Let $C_\gamma(\ve{u})$ denote the set of values at the
moment(s) in the sequence $\ve{u}$ when $\gamma$ is checked, i.e., when
the operation $\ac{ic}$ occurs. For example,
$C_\gamma(\ac{icicricricicic}) = \set{2,1,3}$ and
$C_\gamma(\ac{icicricic}) = \set{2}$.

The \emph{distance objective} is
$\dist = \tup{\set{\ac{ic},\ac{r},\varepsilon}, \mathrm{cost}_{\dist},
  \min}$
and we assume that $\dist$ only has a single counter, say $\gamma$. The
function $\mathrm{cost}_{\dist}$ maps a sequence $\ve{u}$ of counter
actions over $\set{\ac{ic},\ac{r},\varepsilon}$ to
$\sup C_\gamma(\ve{u})$. We will be interested in an objective that
combines $\dist$ with the parity acceptance condition
$\parity$. Formally, given two objectives
$O_1 = \tup{\act_1,f_1,\goal}, O_2 = \tup{\act_2,f_2,\goal}$, we denote
by $O_1 \land O_2$ the objective
$\tup{\act_1 \times \act_2, \max\set{f_1,f_2}, \min}$ if $\goal = \min$,
and $\tup{\act_1 \times \act_2, \min\set{f_1,f_2}, \max}$ if
$\goal = \max$. Thus, for the $\dist \land \parity$-objective, Eve's
goal is to minimize the counter value of $\gamma$ \emph{and} to satisfy
the parity condition at the same time.

\begin{remark}
  Usually, one defines the $\dist$ objective as a special instance of
  the more general \emph{$\mathrm{B}$-objective}. The
  $\mathrm{B}$-objective is minimizing for Eve and may contain multiple
  counters instead of a single one. Moreover, in this objective, the
  counter action $\ac{ic}$ is separated into $\ac{i}$ and $\ac{c}$,
  i.e., the counters may be incremented but not checked. We refer the
  reader to~\cite{BeCCB15} for more details.
\end{remark}

\subsection{Proof of \Cref{lem:costata}}

Let $Q = \tup{\sche{S},\ont{O},G}$ be an OMQ from
$\tup{\class{G},\class{AQ}_0}$. Consider a consistent
$\Gamma_{\sche{S}}$-labeled tree $t$. We are going to devise a
$\dist \land \parity$-automaton $\ata{H}_Q$ over $\Gamma_{\sche{S}}$
such that Eve has an $n$-winning strategy in $\ca{G}(\ata{H}_Q,t)$ if and
only if there is a derivation tree $\ca{T}$ for $\dec{t}$ and $Q$ of
height at most $n$.

Let
$\ata{H}_Q = \tup{S,\Gamma_{\sche{S}},s_0,\set{0,\updownarrow}, \dist
  \land \parity, \delta}$.
For the parity condition, we shall only use the priorities
$\set{0,1}$. The remaining components of $\ata{H}_Q$ are specified in
the following.

\medskip
\noindent
\underline{\smash{\textit{The state set $S$.}}} Let $U_{\sche{S}}$ be the finite set of
constants that is used for arguments in $\Gamma_{\sche{S}}$. The state
set $S$ consists of all atomic formulas $R(a_1,\ldots,a_k)$, where
$R/n \in \sche{S} \cup \sig{\ont{O}}$ and
$a_1,\ldots,a_k \in U_{\sche{S}}$. We set the initial state $q_0$ to
equal $G$. For technical reasons, we include an additional sink
state denoted $\mathsf{sink}$.

\medskip
\noindent
\underline{\smash{\textit{The transition function $\delta$.}}} We define $\delta$ as
follows. Consider a symbol $\rho \in \Gamma_{\sche{S}}$. Firstly, we set
\begin{align*}
 \delta(\mathsf{sink},\rho) &= \ndia{0}\tup{\mathsf{sink}, \varepsilon, 0}.
\end{align*}
Secondly, let $R(a_1,\ldots,a_k)$ be a state different from
$\mathsf{sink}$. We set $\ve{a} = a_1,\ldots,a_k$ and
distinguish cases:
\begin{enumerate}[label={(C\arabic*)}]
\item If $\set{a_1,\ldots,a_k} \not\subseteq \names{\rho}$ then
  \begin{align*}
  \delta(R(\ve{a}),\rho) = \ndia{0}\tup{R(\ve{a}), \ac{ic}, 1}.
  \end{align*}

  In this case, Eve will lose the game as she loops in this state
  $R(\ve{a})$ while incrementing the counter and producing an infinite
  run whose maximum priority that occurs infinitely often (i.e., the
  priority $1$) is odd.

\item Otherwise, if $\rho \models R(\ve{a})$ then
  \begin{align*}
    \delta(R(\ve{a}), \rho) &= \ndia{0}\tup{\mathsf{sink}, \varepsilon, 1}.
  \end{align*}
  In this case, Eve will win the game as she first changes to the sink
  state and then she loops in this state while not increasing the
  counter. In the sink state, she produces an infinite play whose
  maximum priority that occurs infinitely often (the priority $0$) is
  even.

\item Otherwise, let
\begin{align*}
  \tau_1 \colon \alpha_{1,1} \land \cdots \land \alpha_{1,m_1}\ \ldots \ \tau_l \colon \alpha_{l,1} \land \cdots \land \alpha_{l,m_l}
\end{align*}
enumerate all guarded conjunctions of atomic facts from $S$ such that
$\tup{\set{\alpha_{i,1}, \ldots,\alpha_{i,m_i}},\ont{O}}
\models R(\ve{a})$, for all $i = 1,\ldots,l$. We let
\begin{align*}
  \delta(R(\ve{a}),\rho) = \bigvee_{i = 1}^l \bigwedge_{j = 1}^{m_i}\ndia{0}\tup{\alpha_{i,j}, \ac{ic}, 1} \lor
  \ndia{\updownarrow}\tup{R(\ve{a}), \ac{\varepsilon}, 1}.
\end{align*}
Eve may choose between two possibilities here. Either she moves to some
neighboring node in the tree while remaining in state $R(\ve{a})$, or
she may decide to pick a guarded conjunction
$\tau_i \colon \alpha_{i,1} \land \cdots \land \alpha_{i,m_i}$. In the
latter case, Adam challenges Eve's choice by changing the state to one
of the $\alpha_{i,j}$ while incrementing the counter. 
Notice that this case corresponds to the unfolding of a (series of) rules
and thus to the built-up of a derivation tree.
\end{enumerate}
This completes the construction of $\ata{H}_Q$. We briefly comment on
the size of $\ata{H}_Q$ and the time required to construct the same. 

It is clear that the number of states of $\ata{H}_Q$ is exponential in
$\width{\sche{S}}$ and linear in $|\sche{S} \cup
\sig{\ont{O}}|$.
Moreover, the overall construction of $\ata{H}_Q$ takes double
exponential time in the size of $Q$. The determining factor for this
upper bound is the construction of $\delta(\cdot,\cdot)$; more
specifically, the case of condition (C3). Up to logical equivalence,
there are at most double exponentially many conjunctions of the form
$\tau_i \colon \alpha_{i,1} \land \cdots \land \alpha_{i, m_i}$ that
imply a given atomic fact under $\ont{O}$ and $\tau_i$ is of at most
exponential size. Moreover, checking whether an atomic fact is implied
by a database and a set of guarded rules is feasible in \twoexp\ in
combined complexity, and in \textsc{PTime} in data
complexity. Therefore, the transition function can in total be
constructed in \twoexp.

It remains to be shown that $\ata{H}_Q$ is correct, that is:

\begin{lemma}
  \label{lem:costcorrect}
  Suppose ${t}$ is a consistent $\Gamma_{\sche{S}}$-labeled
  tree. Eve has an $n$-winning strategy in $\ca{G}(\ata{H}_Q,t)$ iff
  there is a derivation tree $\ca{T}$ for $\dec{t}$ and $Q$ of height at
  most $n$.
\end{lemma}
\begin{proof}[sketch]
  Suppose first that there is such a $\ca{T}$ of height $n_0 \leq n$. We
  can assume without loss of generality that only the leaf nodes of
  $\ca{T}$ of the form $\beta(\ve{a})$ satisfy
  $\dec{t} \models \beta(\ve{a})$; otherwise, we can simply truncate
  $\ca{T}$.

  Our strategy for Eve will be chosen such that any atomic formula
  $R(a_1,\ldots,a_k)$ that appears as a state in a play consistent with
  that strategy occurs in the label of some node of $\ca{T}$. This is
  trivially satisfied for the initial state, since the root of $\ca{T}$
  is labeled with $G$. Suppose now that the game is at position
  $\tup{\chi, v}$, where $v \in \dom{t}$ is a node of the input tree $t$
  and $\chi$ is of the form $\ndia{0}\tup{\beta(\ve{a}), \ac{ic}, 1}$,
  with $\beta(\ve{a})$ a state of $\ata{H}_Q$. Eve's task is to show
  that she can either match the atom $\beta(\ve{a})$ to the input tree,
  or she proceeds according to (C3) by finding a guarded conjunction of
  facts that imply $\beta(\ve{a})$ under $\ont{O}$. Eve thus proceeds as
  follows:
\begin{itemize}
\item If $t(v) \models \beta(\ve{a})$ then Eve proceeds according to
  (C2). In fact, in this case she has no other choice to do so and, as
  explained in the definition of (C2), such a play will be winning for
  Eve.
\item Otherwise, it must be the case that there is a non-leaf node in
  $\ca{T}$ with a label $\beta(\ve{a})$. Suppose the children of that
  node are labeled with
  $\alpha_1(\ve{a}_1),\ldots,\alpha_k(\ve{a}_k)$. As said, we can
  restrict ourselves to the case where the conjunction
  $\alpha_1(\ve{a}_1) \land \cdots \land \alpha_k(\ve{a}_k)$ is guarded
  and thus, say, $\alpha_1(\ve{a}_1)$ is its guard. Eve first navigates
  to a node $w \in \dom{t}$ of the input tree whose names comprise all
  of $\ve{a}_1$ while remaining in state $\beta(\ve{a})$.  Notice that
  this is possible, since all of $\ve{a}$ are contained as names in the
  nodes on the unique path between $v$ and $w$.  When Eve arrives at
  node $w$, she chooses to challenge Adam by selecting the conjunction
  $\ndia{0}\tup{\alpha_1(\ve{a}_1), \ac{ic},1} \land \cdots \land
  \ndia{0}\tup{\alpha_k(\ve{a}_k), \ac{ic},1}$.
  Adam then selects an arbitrary conjunct
  $\ndia{0}\tup{\alpha_i(\ve{a}_i),\ac{ic},1}$ and the game continues
  from the according position (it is then again Eve's turn).
\end{itemize}
It is easy to see that the choices of Eve that are dictated by $\ca{T}$
lead to an infinite play that satisfies the parity condition. Concerning
the counter, it is only increased when either Eve chooses to ``unfold''
a rule according to (C3). Moreover, the counter is not increased when
she navigates between nodes in the input tree. Therefore, any play
consistent with Eve's strategy has cost at most $n_0 \leq n$.

Conversely, suppose now that Eve has an $n$-winning strategy $\xi$ in
$\ca{G}(\ata{H}_Q,t)$. Let $\pi$ be a play of maximum cost that is
consistent with Eve's strategy. We shall construct a derivation tree
$\ca{T}$ for $\dec{t}$ and $Q$ whose height is bounded by $n$. Let
$\pi = \tup{\chi_0,v_0},\tup{\chi_1,v_1},\ldots,\tup{\chi_k,v_k},\ldots$
be a play consistent with $\xi$. Consider the sequence
$\ve{\chi} = \chi_0,\chi_1,\ldots$ and let
$\alpha_0,\alpha_1,\ldots,\alpha_k,\ldots$ enumerate the states that
appear in atomic formulas in $\ve{\chi}$ such that $\alpha_i$ appears in
$\ve{\chi}$ before $\alpha_{i+1}$. That is, $\alpha_0 = G$ and each
$\alpha_i$ ($i \geq 1$) is either the sink state or results from a
challenge by Adam according to (C3). Notice that, by construction of
$\ata{H}_Q$, if $\alpha_i = \mathsf{sink}$ then
$\alpha_{j} = \mathsf{sink}$ for all $j \geq i$.

We inductively construct a sequence of trees $\ca{T}_0,\ca{T}_1,\ldots,$
such that $\ca{T}_{k+1}$ extends $\ca{T}_{k}$ and there is an $m \geq 0$
such that $\ca{T}_m = \ca{T}_{k}$, for all $k \geq m$. $\ca{T}_m$ will
be a tree that can be extended to derivation tree for $\dec{t}$ and
$Q$. Let $\ca{T}_0$ be the tree with only the root node that is labeled
with $G$. Assume that $\ca{T}_k$ has been constructed. $\ca{T}_{k+1}$ is
defined according to $\pi$:
\begin{itemize}
\item If $\alpha_{k + 1} = \mathsf{sink}$ or $\alpha_{k+1} = \alpha_k$
  then we set $\ca{T}_{k+1} = \ca{T}_k$.
\item If $\alpha_{k+1}$ is an atomic fact chosen by Adam according to
  (C3) in response to Eve's choice of a guarded conjunction
  $\beta_1 \land \cdots \land \beta_l$ that imply an atom $\alpha_i$
  ($i \leq k$), then $\ca{T}_{k+1}$ is obtained from $\ca{T}_k$ by
  adding children $v_1,\ldots,v_l$ with the respective labels
  $\beta_1,\ldots,\beta_l$ to a leaf node $v$ of $\ca{T}_k$ whose label
  is an atomic fact of the form $\alpha_i$. Notice that such a leaf node
  exists, since the state $\alpha_{k+1}$ can only be assumed via the
  existence of such an $\alpha_i$.
\end{itemize}
Observe that $\pi$ must assume the sink state at some point, since $\pi$
is winning for Eve and, by construction, it loops in this state
(otherwise the parity condition would not be satisfied). Thus, there is
an $m \geq 0$ such that $\ca{T}_m = \ca{T}_k$ for all $k \geq m$. It
remains to be shown that $\ca{T}_m$ can be extended to a derivation tree
$\ca{T}$ for $\dec{t}$ and $Q$. Recall that we chose $\pi$ to be a play
consistent to Eve's strategy that is of maximal cost. Clearly, the cost
of $\pi$ equals the height of $\ca{T}_m$. Roughly, $\ca{T}_m$ consists
of a finite branch starting at the root (which is labeled with $G$)
whose leaf node is labeled by an atom from $\dec{t}$. Moreover, each
branch may have children that may not be database atoms of
$\dec{t}$. We can, however, easily check that we can attach subtrees to
these ``incomplete'' nodes such that the resulting tree $\ca{T}$ becomes
a derivation tree for $\dec{t}$ and $Q$. If this was not possible, Adam
could find a play which forces Eve to lose. Moreover, the height of
$\ca{T}$ must equal the height of $\ca{T}_m$, since otherwise Adam could
find a play that has higher cost than $\pi$, which is impossible due to
our choice of $\pi$. Notice that, by construction, the cost of $\pi$
equals the height of $\ca{T}$ and $\height{\ca{T}} \leq n$.
\end{proof}

Hence, by~\Cref{lem:costcorrect}, we know that, for any consistent $\Gamma_{\sche{S}}$-labeled tree
$t$, $\val{\ata{H}_Q}(t) = n$ if and only if $n$ is the minimal $n_0$
such that there is a derivation tree of height $n_0$ for $\dec{t}$ and
$Q$. Therefore, for all $n \in \mathbb{N}$, $\val{\ata{H}_Q}(t) = n$ iff
$\cost{\dec{t}}{Q} = n$. We thus obtain that $\val{\ata{H}_Q}$ is
bounded iff $\costq{Q}$ is finite. This concludes the proof of
\Cref{lem:costata}.

\section{Proofs for \Cref{sec:treeification}}

In this section, we also consider CQs that contain \emph{equality atoms}
of the form $x = y$ in their bodies. Notice that, in non-empty CQs,
these can always be removed by appropriately identifying variables. We
allow such atoms, since the results we rely upon explicitly make use of
such atoms.

We say that a CQ $q(\ve{x})$ is \emph{answer-guarded} if it contains an
atom in its body that has all answer variables of $q(\ve{x})$ as
arguments. Notice that every (non-empty) Boolean CQ is trivially
answer-guarded.\footnote{Note that we can view the empty CQ, denoted
  $\top$, also as answer-guarded since it is equivalent to
  $\exists x\, x = x$.} Also notice that the body of any
frontier-guarded rule can be seen as an answer-guarded CQ.

\subsection{Preliminaries: Treeification}

\underline{\smash{\textit{Strictly acyclic queries.}}} Let $q(\ve{x})$
be an answer-guarded CQ over a schema $\sche{S}$. We say that $q(\ve{x})$ is
\emph{acyclic}, if there is a tree decomposition
$\delta = \tup{\ca{T},\tup{X_v}_{v \in T}}$ of $q(\ve{x})$ such that,
for all $v \in T$, there is an atom $\alpha$ of $q(\ve{x})$ such that
$X_v \subseteq \var{\alpha}$. If there is such a $\delta$ that, in
addition, has a bag containing all answer variables of $q(\ve{x})$, then
we say that $q$ is \emph{strictly acyclic}.

A \emph{guarded formula} is a first-order formula where each occurrence
of a quantifier is of either forms
\begin{align*}
  \forall \ve{y}\,(\alpha(\ve{x},\ve{y}) \limpl \psi)\quad\text{or}\quad\exists\ve{y}\,(\alpha(\ve{x},\ve{y}) \land \psi),
\end{align*}
where $\alpha$ is an atomic formula, called \emph{guard}, and all the
free variables of $\psi$ (denoted $\free{\psi}$) are contained in
$\ve{x} \cup \ve{y}$. We also permit that the guard is an equality atom
of the form $x = y$. In the following, we are interested in guarded
formulas that contain only existential quantification and conjunction,
and we restrict ourselves to those in the remainder of this paper. We
say that a guarded formula $\varphi(\ve{x})$ is \emph{strictly guarded},
if it is of the form
$\exists\ve{y}\,(\alpha(\ve{x},\ve{y}) \land \psi)$, i.e., all free
variables in a strictly guarded formula are covered by an atom as well.
It is well-known that every strictly acyclic formula is equivalent to a
strictly guarded formula and vice versa (see~\cite{FlFG02}).

\medskip
\noindent
\underline{\smash{\textit{Treeifying CQs.}}}  Given an answer-guarded CQ
$q(\ve{x})$ over $\sche{S}$ and a schema $\sche{T} \supseteq \sche{S}$,
the \emph{$\sche{T}$-treeification of $q(\ve{x})$} is the set
$\Lambda^{\sche{T}}_q$ of all strictly acyclic CQs $q'(\ve{x})$ over
$\sche{T}$ such that
\begin{enumerate*}[label={(\roman*)}]
\item $q'$ is \emph{contained} in $q$, in symbols $q' \subseteq q$, that
  is, for any $\sche{T}$-database $\db{D}$, if $\db{D} \models q'$ then
  also $\db{D} \models q$, and
\item $q'$ is minimal in the sense that removing one atom from $q'$
  turns $q'$ into a CQ that is either not strictly acyclic or that is
  not contained in $q$ anymore.
\end{enumerate*}

It can be shown that all the CQs contained in $\Lambda^{\sche{T}}_q$ can
be restricted as to contain only CQs of size at most $3|q|$, where $|q|$
denotes the number of atoms in $q$. Hence, $\Lambda^{\sche{T}}_q$ can be
seen as a UCQ that is of exponential size in the size of $q$. Notice
that $q(\ve{x})$ is in general not equivalent to its
treeification. However, $q(\ve{x})$ and $\Lambda^{\sche{T}}_q$ are
equivalent over acyclic $\sche{T}$-databases (acyclicity for databases
is defined as for CQs)~\cite{BaCS15}.

\medskip
\noindent
\underline{\smash{\textit{Treeifying OMQs from
      $\tup{\class{FG},\class{AQ}_0}$.}}} Let
$Q = \tup{\sche{S},\ont{O},G}$ from $\tup{\class{FG},\class{AQ}_0}$. The
\emph{width} of $\ont{O}$, denoted $\width{\ont{O}}$, is the maximum
number of variables that appear in any body of a rule from
$\ont{O}$. Fix a new relation symbol $C$ of arity $\width{\ont{O}}$. 

We are now going to describe a translation $\translation{Q}{C}$ in full
detail that takes $Q$ and transforms it into an OMQ $\translation{Q}{C}$
from $\tup{\class{G},\class{AQ}_0}$ with data schema
$\sche{S} \cup \set{C}$. Firstly, we set
\begin{align*}
  \translation{Q}{C} = \left(\sche{S} \cup \set{C}, \bigcup_{\tau \in \ont{O}} \eta^{\sche{S} \cup \sig{\ont{O}}}_C(\tau), G\right),
\end{align*}
where the definition of $\eta^{\sche{T}}_C(\tau)$, for $\tau \in \ont{O}$
and a schema $\sche{T}$, is as follows. Suppose $\tau$ is of the form
$\varphi(\ve{x},\ve{z}) \limpl
\exists\ve{y}\,\beta(\ve{x},\ve{y})$. Then we set
\begin{align*}
  f_C^{\sche{T}}(\tau) = \left\{q(\ve{x}) \limpl \exists \ve{y}\,\beta(\ve{x},\ve{y}) \mid q(\ve{x}) \in \Lambda^{\sche{T} \cup \set{C}}_{\exists\ve{z}\,\varphi(\ve{x},\ve{z})}\right\}.
\end{align*}
Notice though, strictly speaking, the rules
$q(\ve{x}) \limpl \exists \ve{y}\,\beta(\ve{x},\ve{y})$ may not be
guarded. However, since $q(\ve{x})$ is strictly acyclic, we may unfold
$q(\ve{x}) \limpl \exists\ve{y}\,\beta(\ve{x},\ve{y})$ into linearly
many guarded rules by using additional auxiliary predicates. The result
of this unfolding will be denoted $ \eta^{\sche{T}}_C(\tau)$.

We are going to describe this unfolding step in more detail in the
following. Let $\chi(\ve{x})$ be a strictly guarded formula
equivalent to $q(\ve{x})$. During the unfolding step, we are going to
introduce fresh auxiliary predicates of the form $T_{\eta}/k$, where
$\eta$ is a subformula of $\varphi(\ve{x})$ and $k$ is the number of
free variables of $\eta$. We shall treat this predicates modulo logical
equivalence, i.e., we set $T_{\eta_1} = T_{\eta_2}$ iff
$\eta_1 \equiv \eta_2$. We unfold the query $q(\ve{x})$ inductively
according to the construction of $\chi(\ve{x})$.

Suppose first that $\chi(\ve{x}) \equiv \beta(\ve{x})$ for some
relational atom $\beta(\ve{x})$. We translate $q(\ve{x})$ to the
rule
\begin{align*}
  \beta(\ve{x}) &\limpl T_{\chi(\ve{x})}(\ve{x}).
\end{align*}
Suppose now that
$\chi(\ve{x}) \equiv \exists\ve{y}\,(\gamma(\ve{x},\ve{y}) \land \eta)$,
where $\free{\eta} \subseteq \ve{x} \cup \ve{y}$ and $\gamma(\ve{x},\ve{y})$
is a relational atom.\footnote{The case where the guard is actually an
  equality atom $x = y$ is handled in a similar fashion just by
  identifying variables in the resulting rule.} Let
$\eta_1(\ve{x}_1),\ldots,\eta_k(\ve{x}_k)$ be strictly guarded formulas
such that $\free{\eta_i} = \set{\ve{x}_i}$ (for $i =1 ,\ldots,k$) and
whose conjunction is equivalent to $\eta$. Then we rewrite
$q(\ve{x}) \limpl \exists\ve{y}\,\beta(\ve{x},\ve{y})$ into the rule
\begin{align*}
  \gamma(\ve{x},\ve{y}), T_{\eta_1}(\ve{x}_1),\ldots,T_{\eta_k}(\ve{x}_k) &\limpl T_{\chi(\ve{x})}(\ve{x}),
\end{align*}
and, in addition, add the according translations for the formulas
$\eta_1(\ve{x}_1),\ldots,\eta_k(\ve{x}_k)$.

The \emph{unfolding}\footnote{Of course, the notion of unfolding depends
  on the choice of $\chi(\ve{x})$. However, it is easily seen that we
  arrive at an equivalent set of rules, no matter which $\chi(\ve{x})$
  equivalent to $q(\ve{x})$ is chosen.} of the rule
$q(\ve{x}) \limpl \exists \ve{y}\,\beta(\ve{x},\ve{y})$ is then the set
of rules resulting from translating $\chi(\ve{x})$ plus the rule
\begin{align*}
  T_{\chi(\ve{x})} \limpl \exists \ve{y}\,\beta(\ve{x},\ve{y}). 
\end{align*}
As mentioned above, we set
\begin{align*}
  \eta_C^{\sche{T}}(\tau) = \set{\sigma \mid \text{$\sigma$ is a rule contained in the unfolding of $\tau$}}.
\end{align*}
It is easy to see that the unfolding introduces at most linearly many
new auxiliary predicates per rule. Thus, in total, the number of rules
contained in $\translation{\ont{O}}{C}$ is exponential in the number of
rules of $\ont{O}$ and we may introduce an exponential number of new
auxiliary relation symbols.

\begin{example}
  We provide an example due to~\cite{ourpods} that exemplifies the unfolding of rules as
  described above. Suppose
  $q(x) = \exists y, z\,(R(x,y) \land S(x,x) \land R(y,z))$ which is
  equivalent to the strictly guarded formula
  $\chi(x) = S(x,x) \land \exists y\,(R(x,y) \land
  \exists z\,R(y,z))$. Suppose we want to unfold the frontier-guarded
  rule $q(x) \limpl O(x)$, where $O/1$ is a unary relation symbol. Then
  the unfolding described above yields the set of rules
  \begin{align*}
    T_{\chi(x)}(x) &\limpl O(x),\\
    S(x,x), T_{\exists y (R(x,y) \land \exists z R(y,z))}(x) &\limpl T_{\chi(x)}(x),\\
    R(x,y), T_{\exists z R(y,z)}(y) &\limpl T_{\exists y (R(x,y) \land \exists z R(x,z))}(x),\\
    R(y,z) &\limpl T_{\exists z R(y,z)}(y).
  \end{align*}
  Notice that here we somehow did not pedantically follow the exact
  translation, since we treated $\exists z\,R(y,z)$ as a strictly
  guarded formula, thereby invoking the fact that it is equivalent to
  $\exists z\,(R(y,z) \land R(y,z))$.\hfill\markfull
\end{example}

In~\cite{ourpods}, the following lemmas are shown:

\begin{lemma}
  \label{lem:equivacyclic}
  Suppose $\db{D}$ is an acyclic $(\sche{S} \cup \set{C})$-database. Then
  $\db{D} \models Q$ iff $\db{D} \models \translation{Q}{C}$.
\end{lemma}

\begin{lemma}
  \label{lem:treemodelfg}
  If $\db{D} \models Q$ then there is an $\sche{S}$-database
  $\db{D}^{\ast}$ of tree-width at most $\max\set{0,\width{\ont{O}} - 1}$ such that
  \begin{enumerate}
    \item $\db{D}^\ast \models Q$,
    \item there is a homomorphism from $\db{D}^\ast$ to $\db{D}$.
  \end{enumerate}
  Moreover, if $q$ is a Boolean CQ of tree-width at most $w \geq 0$ and
  $\db{D} \models q$, then there is an $\sche{S}$-database
  $\db{D}^{\ast}$ of tree-width at most $w$ such that
  \begin{enumerate}
    \item $\db{D}^\ast \models q$,
    \item there is a homomorphism from $\db{D}^\ast$ to $\db{D}$.
  \end{enumerate}
\end{lemma}

\begin{lemma}
  \label{lem:acymodel}
  If $Q \in \tup{\class{G},\class{AQ}_0}$ is an OMQ with data schema
  $\sche{S}$ and $\db{D} \models Q$, then there is an acyclic
  $\db{D}^\ast$ such that
  \begin{enumerate}
    \item $\db{D}^\ast \models Q$ and
    \item there is a homomorphism from $\db{D}^\ast$ to $\db{D}$.
  \end{enumerate}
  Moreover, an according statement holds for acyclic Boolean CQs as well. 
\end{lemma}

The following lemma will be the main ingredient towards a proof of
\Cref{the:main-result}:

\begin{lemma}
  \label{lem:foequiv}
  $Q$ is FO-rewritable iff $\translation{Q}{C}$ is.
\end{lemma}
\begin{proof}
  Throughout the proof, let
  $w = \max\set{0,\width{\ont{O}}-1}$. 

  Suppose first that $Q$ is FO-rewritable and let
  $q = \bigvee_{i = 1}^n p_i$ be a UCQ equivalent to $Q$.

  \begin{claim}
    The UCQ $q$ is equivalent to a UCQ whose disjuncts are all of tree-width at
    most $w$.
  \end{claim}
  \begin{proof}
    Let $q'$ be a UCQ that contains a disjunct $p'$ iff
    \begin{enumerate*}[label={(\roman*)}]
    \item $p'$ has tree-width at most $w$,
    \item $p' \subseteq p_i$ for some $i = 1,\ldots,n$, and
    \item $p'$ is minimal with respect to these properties.
  \end{enumerate*}
  Moreover, we require from that no disjunct in $q'$ is homomorphically
  equivalent to another disjunct of $q'$. Thus, $q'$ is indeed a finite
  UCQ. We show that $q$ is equivalent to $q'$.

  Suppose first that $\db{D} \models q$. Then also $\db{D} \models Q$
  and, by \Cref{lem:treemodelfg}, there is an $\sche{S}$-database
  $\db{D}^\ast$ of tree-width at most $w$ such that
  $\db{D}^\ast \models Q$ and $\db{D}^\ast$ maps to $\db{D}$. Thus, also
  $\db{D}^\ast \models q$. Since $\db{D}^\ast$ has tree-width at most
  $w$, there is a disjunct $p$ in $q'$ and a
  $\db{D}' \subseteq \db{D}^\ast$ such that $p$ is homomorphically
  equivalent to $\db{D}'$. Thus $\db{D}' \models p$ and so
  $\db{D}^\ast \models p$. Since $\db{D}^\ast$ maps homomorphically to
  $\db{D}$, we have $\db{D} \models p$ and so $\db{D} \models q'$
  follows.

  Suppose now that $\db{D} \models q'$, i.e., $\db{D} \models p$ for
  some disjunct $p$ of $q'$. Since $p \subseteq p_i$ for some
  $i = 1,\ldots,n$, there is a homomorphism from $p_i$ to $p$. Hence, it
  follows that $\db{D} \models p_i$ and thus $\db{D} \models q$.
  \end{proof}

  Suppose now that each $p_i$ ($i = 1,\ldots,n$) has tree-width at most
  $w$.  Let $\delta_i = \tup{\ca{T}_i, \tup{X_{i,v}}_{v \in T_i}}$ be a
  tree decomposition of $p_i$ of width at most $w$. A \emph{variant} of
  $p_i$ is a CQ $p$ over $(\sche{S} \cup \set{C})$ that
  \begin{enumerate*}[label={(\roman*)}]
  \item results from $p_i$ by adding a set of atoms of the form
    $C(x_0,\ldots,x_w)$ with $\set{x_0,\ldots,x_w} \subseteq X_{i,v}$
    for some $v \in T_i$, and
  \item is acyclic.
  \end{enumerate*}
  We let $p_i'$ be the UCQ over $(\sche{S} \cup \set{C})$ that contains
  a disjunct for each variant of $p_i$. Moreover, we let
  $q' = \bigvee_{i = 1}^n p_i'$ and assume again that $q'$
  contains no two distinct disjuncts that are homomorphically
  equivalent. Obviously, $q'$ is a finite UCQ, and we claim that $q'$ is
  a UCQ-rewriting of $\translation{Q}{C}$.

  Indeed, suppose $\db{D}$ is an $(\sche{S} \cup \set{C})$-database such
  that $\db{D} \models \translation{Q}{C}$. According to
  \Cref{lem:acymodel}, there is an acyclic
  $(\sche{S} \cup \set{C})$-database $\db{D}^\ast$ such that
  $\db{D}^\ast \models \translation{Q}{C}$ and such that $\db{D}^\ast$
  maps homomorphically to $\db{D}$. By \Cref{lem:equivacyclic} we have
  $\db{D}^\ast \models Q$ as well. Let $\db{D}^\ast[\sche{S}]$ denote
  the database $\db{D}^\ast$ restricted to $\sche{S}$. Since $C$ does
  not appear in $Q$, we must have $\db{D}^\ast[\sche{S}] \models Q$ as
  well and so $\db{D}^\ast[\sche{S}] \models p_i$ for some
  $i = 1,\ldots,n$. It is now easy to check that there is a variant $p$
  of $p_i$ such that $\db{D}^\ast \models p$. Hence,
  $\db{D}^\ast \models q'$ and so $\db{D} \models q'$ as required.

  Conversely, suppose that $\db{D} \models q'$, i.e., $\db{D} \models p$
  for some CQ $p$ that is a variant of some $p_i$. By
  \Cref{lem:acymodel} there is an acyclic
  $(\sche{S} \cup \set{C})$-database $\db{D}^\ast$ such that
  $\db{D}^\ast \models p$ and $\db{D}^\ast$ homomorphically maps to
  $\db{D}$. Obviously, there is a homomorphism from $p_i$ to $p$, since
  $p$ is a variant and results from $p_i$ just by adding atoms. Hence,
  also $\db{D}^\ast \models p_i$ and thus $\db{D}^\ast \models q$
  follows. We then obtain $\db{D}^\ast \models Q$ and by
  \Cref{lem:equivacyclic} also $\db{D}^\ast \models \translation{Q}{C}$.
  Since $\translation{Q}{C}$ is closed under homomorphisms,
  $\db{D} \models \translation{Q}{C}$ follows.

  Suppose now that $\translation{Q}{C}$ is FO-rewritable and let
  $q = \bigvee_{i = 1}^n p_i$ be a UCQ equivalent to
  $\translation{Q}{C}$. We show that $Q$ is FO-rewritable as well. In
  this case, we can assume that $q$ is actually a disjunction of acyclic
  CQs, a proof of this fact can be obtained similarly to the claim above
  and is left to the reader. Let $p_i'$ be the CQ that results from
  $p_i$ by dropping all atoms of the form $C(x_0,\ldots,x_w)$. Moreover,
  let $q' = \bigvee_{i=1}^n p_i'$. We claim that $q'$ is a UCQ
  equivalent to $Q$. 

  Suppose first that $\db{D} \models Q$. By \Cref{lem:treemodelfg} there
  is an $\sche{S}$-database $\db{D}^\ast$ of tree-width at most $w$ such
  that $\db{D}^\ast \models Q$ and $\db{D}^\ast$ homomorphically maps to
  $\db{D}$. Fix a tree decomposition
  $\delta^\ast = \tup{\ca{T},\tup{X_v}_{v \in T}}$ of $\db{D}^\ast$. We
  can turn $\db{D}^\ast$ into an acyclic
  $(\sche{S} \cup \set{C})$-database by adding to $\db{D}^\ast $ all
  facts of the form $C(a_0,\ldots,a_w)$ such that
  $\set{a_0,\ldots,a_w} \subseteq X_v$ for some $v \in T$. Call the
  resulting database $\db{D}'$. Obviously, $\db{D}' \models Q$ and since
  $\db{D}'$ is acyclic, we obtain $\db{D}' \models \translation{Q}{C}$
  by \Cref{lem:equivacyclic}. Therefore, $\db{D}' \models p_i$ for some
  $i = 1,\ldots,n$. Since $p_i'$ contains no atoms of the form
  $C(x_0,\ldots,x_w)$, it follows that $\db{D}' \models p_i'$ as well
  and so $\db{D}^\ast \models p_i'$. Thus $\db{D} \models p_i'$ and so
  $\db{D} \models q'$ as required.

  Conversely, suppose now that $\db{D} \models q'$, i.e.,
  $\db{D} \models p_i'$ for some $i = 1,\ldots,n$. Notice that $p_i'$
  has tree-width at most $w$ by construction. Using
  \Cref{lem:treemodelfg}, we infer that there is an $\sche{S}$-database
  $\db{D}^\ast$ such that $\db{D}^\ast \models p_i'$ and such that there
  is a homomorphism $h$ from $\db{D}^\ast$ to $\db{D}$. Now fix a tree
  decomposition $\delta_i = \tup{\ca{T}, \tup{X_v}_{v \in T}}$ of $p_i'$
  of width at most $w$. We can see $\delta_i$ also as a tree
  decomposition of $p_i$ that witnesses that $p_i$ is acyclic. Now we
  extend $\db{D}^\ast$ to an $(\sche{S} \cup \set{C})$-database as
  follows. Suppose $C(x_0,\ldots,x_w)$ occurs in $p_i$ but has been
  deleted from $p_i'$. Then $\set{x_0,\ldots,x_w} \subseteq X_v$ for
  some $v \in T$. We can assume w.l.o.g~that
  $\dom{h} \cap \set{x_0,\ldots,x_w} \neq \emptyset$; otherwise we can
  drop that atom from $p_i$. Pick a
  $y \in \dom{h} \cap \set{x_0,\ldots,x_w}$.  Now we add to
  $\db{D}^\ast$ the atom $C(a_0,\ldots,a_w)$, where
  $a_i =  h(x_i)$ if $x_i \in \dom{h}$ and $a_i = h(y)$
  otherwise. We repeat this construction for all occurrences of an atom
  of the form $C(x_0,\ldots,x_w)$ in $p_i$. Call the resulting database
  $\db{D}'$. It is clear that $h$ is a homomorphism from $p_i$ to
  $\db{D}'$. Notice also that $\db{D}'[\sche{S}] = \db{D}^\ast$. Now
  since $\db{D}' \models p_i$, we must have
  $\db{D}' \models \translation{Q}{C}$ and so by \Cref{lem:acymodel}
  there is an acyclic $(\sche{S} \cup \set{C})$-database $\db{D}''$ such
  that $\db{D}'' \models \translation{Q}{C}$ and $\db{D}''$
  homomorphically maps to $\db{D}'$. By \Cref{lem:equivacyclic} we have
  $\db{D}'' \models Q$ as well and hence also $\db{D}' \models Q$. But
  $\db{D}'[\sche{S}] = \db{D}^\ast$, whence $\db{D}^\ast \models Q$ follows
  since $C$ does not occur in $Q$. Since $\db{D}^\ast$ homomorphically
  maps to $\db{D}$, we obtain $\db{D} \models Q$ as required.
\end{proof}

\subsection{Proof of \Cref{the:main-result}}

\underline{\smash{\textit{Lower bounds.}}} Since, according
to~\cite{BiHLW16}, FO-rewritability for the class
$\tup{\ca{ELI},\class{BCQ}}$ is already hard for \twoexp\ according to
\cite{BiHLW16}, the following hardness results follow immediatlely:
\begin{itemize}
\item \twoexp-hardness for $\forew{\class{C}}{\class{Q}}$ with
  $\class{C} \in \set{\class{G},\class{FG}}$ and
  $\class{Q} \in \set{\class{CQ},\class{UCQ}}$;
\item \twoexp-hardness for $\forew{\class{FG}}{\class{AQ}_0}$.
\end{itemize}
Moreover, in~\cite{BiLW13}, it is shown that $\forew{\ca{ELI}}{AQ_0}$
is \expo-hard. Therefore, for OMQs of bounded arity, \expo-hardness for $\forew{\class{G}}{\class{AQ}_0}$ follows.

The only missing lower bound is therefore the \twoexp\ lower bound for
$\forew{\class{G}}{\class{AQ}_0}$. 

Let $Q_1$ and $Q_2$ be Boolean OMQs with data schema $\sche{S}$. We say
that $Q_1$ is \emph{contained} in $Q_2$, if $\db{D} \models Q_1$ implies
$\db{D} \models Q_2$ for every $\sche{S}$-database $\db{D}$. We are
going to use the following result which is implicit in~\cite{BaRV14}:

\begin{theorem}
\label{thm:pablo}
The problem of deciding whether a OMQ $Q_1 = \tup{\sche{S},\ont{O},G_1}$
from $\tup{\class{G},\class{AQ}_0}$ is contained in an OMQ
$Q_2 = \tup{\sche{S},\ont{O},G_2}$ is hard for \twoexp. This is true
even for the case where $Q_2$ is FO-rewritable.
\end{theorem}

\begin{remark}
  In~\cite{BaRV14}, a slightly different statement is proved. The
  authors prove in fact that deciding whether a \emph{guarded Datalog}
  program is contained in a Boolean acyclic UCQ is hard for
  \twoexp. Guarded Datalog can easily be seen as a fragment of
  $\tup{\class{G},\class{AQ}_0}$. Moreover, a Boolean acyclic UCQ can
  easily be written as an OMQ from $\tup{\class{G},\class{AQ}_0}$
  (cf.~the discussion of ``unfolding'' strictly acyclic queries in the
  definition of treeifications).
\end{remark}

To prove that $\forew{\class{G}}{\class{AQ}_0}$ is hard for \twoexp, we
are going to reduce the problem mentioned in \Cref{thm:pablo} to
$\forew{\class{G}}{\class{AQ}_0}$.

Let $Q_1 = \tup{\sche{S}, \ont{O}_1,G_1}$ and
$Q_2 = \tup{\sche{S},\ont{O}_2, G_2}$ be as in the hypothesis of
\Cref{thm:pablo}. Without loss of generality, we may assume that the
predicates $Q_1$ and $Q_2$ use and that do not appear in $\sche{S}$ are
distinct. We are going to construct an OMQ $Q'$ that falls in
$\tup{\class{G},\class{AQ}_0}$ such that $Q'$ is FO-rewritable iff $Q_1$
is contained in $Q_2$.

Let $Q' = \tup{\sche{S}, \ont{O}',G_2}$, where
\begin{itemize}
  \item $\sche{S}' = \sche{S} \cup \set{R/2, A/1,B/1}$;
  \item $\ont{O}'$ is the union of $\ont{O}_1$ and $\ont{O}_2$ plus the
    rules
  \begin{align*}
    R(x,y), A(y) &\limpl A(x),\\
    A(x),B(x),G_1 &\limpl G_2.
  \end{align*}
\end{itemize}
Notice that $G_2$ is also the query component of $Q'$.

\begin{lemma}
  \label{lem:reduction}
  $Q_1$ is contained in $Q_2$ iff $Q'$ is FO-rewritable.
\end{lemma}
\begin{proof}
  Assume first that $Q_1$ is not contained in $Q_2$. Then there is an
  $\sche{S}$-database $\db{D}$ such that $\db{D} \models Q_1$ and
  $\db{D} \not\models Q_2$. By \Cref{lem:treemodel}, there is a
  $\db{D}^\ast$ of tree-width at most $\max\set{0, \width{\sche{S}}-1}$
  such that $\db{D}^\ast \models Q_1$. Moreover, there also is a
  homomorphism from $\db{D}^\ast$ to $\db{D}$. Since $Q_2$ is closed
  under homomorphisms, we must also have $\db{D}^\ast \not\models Q_2$.
  For each $k > 0$, let $\db{D}_k$ be the $\sche{S}'$-database extending
  $\db{D}^\ast$ with the facts
  \begin{align*}
    B(a_0),R(a_0,a_1),\ldots,R(a_{k-1},a_k),A(a_k),
  \end{align*}
  where $a_0,\ldots,a_k$ do not occur in $\adom{\db{D}^\ast}$. It is
  easy to check that $\db{D}_k \models Q'$ for all $k > 0$. Moreover, no
  proper subset of $\db{D}_k$ satisfies $Q'$. By virtue of
  \Cref{pro:semanticmain}, $Q'$ is thus not FO-rewritable.

  Conversely, suppose that $Q_1$ is contained in $Q_2$. Recall that $Q_2$
  is FO-rewritable and, therefore, there is a UCQ $q$ over $\sche{S}$
  that is equivalent to $Q_2$. We claim that $q$ is a UCQ-rewriting for
  $Q'$ as well. 

  Indeed, suppose first that $\db{D} \models q$ for some
  $\sche{S}'$-database $\sche{S}$. Since $q$ uses only symbols from
  $\sche{S}$, we obtain that $\db{D}[\sche{S}] \models q$ as well. Since
  $q$ is equivalent to $Q_2$, we get $\db{D}[\sche{S}] \models Q_2$ and,
  by construction of $Q'$, so $\db{D} \models Q'$.

  Suppose now that $\db{D} \models Q'$ for some $\sche{S}'$-database
  $\db{D}$. By construction of $Q'$, we must then have
  $\db{D} \models Q_2$ or $\db{D} \models Q_1$. In the former case, we
  are done since $Q_2$ and $q$ are equivalent. In the latter case, we
  get $\db{D}[\sche{S}] \models Q_1$ whence
  $\db{D}[\sche{S}] \models Q_2$ since $Q_1$ is contained in
  $Q_2$. Therefore also $\db{D}[\sche{S}] \models q$ and thus
  $\db{D} \models q$. This proves the claim.
\end{proof}

It is clear that $Q'$ can be constructed from $Q_1$ and $Q_2$ in
polynomial time. Therefore, \twoexp-hardness for
$\forew{\class{G}}{\class{AQ}_0}$ follows by \Cref{lem:reduction}.

\medskip
\noindent
\underline{\smash{\textit{Upper bounds.}}}
We shall now prove that $\forew{\class{FG}}{\class{UCQ}}$ is in
\twoexp. Following a similar result for description logics
in~\cite{BiHLW16}, we first show that we can focus on Boolean UCQs:

\begin{lemma}
  \label{lem:ucqtoubcq}
  Let $\class{C} \in \set{\class{FG},\class{G}}$. Then
  $\forew{\class{C}}{\class{UCQ}}$ can be reduced in polynomial time to
  $\forew{\class{C}}{\class{UBCQ}}$.
\end{lemma}
\begin{proof}[sketch]
  Let $Q = \tup{\sche{S},\ont{O},q(\ve{x})}$ be an OMQ from
  $\tup{\class{C},\class{UCQ}}$ with $\ve{x} = x_1,\ldots,x_n$. We let
  $\sche{S}' = \sche{S} \cup \set{A_1,\ldots,A_n}$, where
  $A_1,\ldots,A_n$ are fresh unary predicates. Let $q'(\ve{x})$ be the
  UCQ that results from $q(\ve{x})$ by adding the conjunction
  $A_1(x_1) \land \cdots \land A_n(x_n)$ to every disjunct of
  $q(\ve{x})$. Let
  $Q' = \tup{\sche{S}',\ont{O},\exists\ve{x}\,q'(\ve{x})}$. It
  is not hard to check that $Q$ is FO-rewritable iff $Q'$ is. 

  Indeed, if
  $\varphi_Q(x_1,\ldots,x_n)$ is an FO-rewriting of $Q$, then
  $\exists x_1,\ldots,x_n\,(\varphi_Q(x_1,\ldots,x_n) \land A_1(x_1)
  \land \cdots \land A_n(x_n))$ is one of $Q'$. 

  Conversely, if $Q'$ is FO-rewritable then there is a Boolean UCQ $p'$
  that is equivalent to $Q'$. Now, for any $\sche{S}'$-database
  $\db{D}$, $\db{D} \models q'$ iff there are
  $a_1,\ldots,a_n \in \adom{\db{D}}$ such that
  $A_1(a_1),\ldots,A_n(a_n) \in \db{D}$ and
  $\db{D} \models Q(a_1,\ldots,a_n)$. Let $p$ be the UCQ that results
  from $p'$ by removing all occurrences of $A_i(x_i)$ and the associated
  existential quantifier $\exists x_i$. It is easy to see that $p$ is a
  UCQ-rewriting of $Q$.
\end{proof}

Now consider an OMQ $Q = \tup{\sche{S},\ont{O},q}$ from
$\tup{\class{FG},\class{UBCQ}}$. In a first step, we transform $Q$ into
an equivalent OMQ $Q'$ that falls in
$\tup{\class{FG},\class{AQ}_0}$. This is easy: we simply choose a fresh
predicate $G$ of arity zero and add to $\ont{O}$ the rules $p \limpl G$
for every disjunct $p$ of $q$. Notice that $\ont{O}$ is still
frontier-guarded, since $q$ is Boolean. Call the resulting ontology
$\ont{O}'$, i.e., $Q' = \tup{\sche{S},\ont{O}', G}$.

Now we choose a fresh predicate $C$ of arity $\width{\ont{O}'}$. We then
construct the OMQ $\translation{Q'}{C}$ that has data schema
$\sche{S} \cup \set{C}$. This translation takes exponential time, and
the ontology of $\translation{Q'}{C}$ may be of exponential
size. However, as already mentioned in the main body of the paper, the
arity of each predicate occurring in $\translation{Q'}{C}$ is at most
$\width{\ont{O}'}$.

The OMQ $\translation{Q'}{C}$ falls in
$\tup{\class{G},\class{AQ}_0}$. We can, according to
\Cref{the:cost-automata-approach}, decide FO-rewritability for that
class in \twoexp, with a double exponential dependence only on the width
of the data schema. Since
$\width{\sche{S} \cup \set{C}} = \width{\ont{O}'}$, it follows that
FO-rewritability of $\translation{Q'}{C}$ can be decided in \twoexp,
where the second exponent of the run-time depends on $\width{\ont{O}'}$
only. Hence, we can decide whether $\translation{Q'}{C}$ if
FO-rewritable in \twoexp. Given that the construction of
$\translation{Q'}{C}$ (starting with $Q$) is, of course, also feasibly
in \twoexp, the fact that $\forew{\class{FG}}{\class{UBCQ}}$ is in
\twoexp\ follows by \Cref{lem:foequiv}. Using \Cref{lem:ucqtoubcq},
we obtain that $\forew{\class{FG}}{\class{UCQ}}$ is in \twoexp\ as well.

%%% Local Variables:
%%% fill-column: 72
%%% TeX-PDF-mode: t
%%% TeX-debug-bad-boxes: t
%%% TeX-master: "ijcai18.tex"
%%% TeX-parse-self: t
%%% TeX-auto-save: t
%%% reftex-plug-into-AUCTeX: t
%%% End:

% \bibliographystyle{named}
% \bibliography{ijcai18}

\begin{thebibliography}{}

\bibitem[\protect\citeauthoryear{Ajtai and Gurevich}{1994}]{AjGu94}
Mikl{\'{o}}s Ajtai and Yuri Gurevich.
\newblock Datalog vs first-order logic.
\newblock {\em J. Comput. Syst. Sci.}, 49(3):562--588, 1994.

\bibitem[\protect\citeauthoryear{Baget \bgroup \em et al.\egroup
  }{2011}]{BLMS11}
Jean-Fran\c{c}ois Baget, Michel Lecl{\`e}re, Marie-Laure Mugnier, and Eric
  Salvat.
\newblock On rules with existential variables: {W}alking the decidability line.
\newblock {\em Artif. Intell.}, 175(9-10):1620--1654, 2011.

\bibitem[\protect\citeauthoryear{B{\'{a}}r{\'{a}}ny \bgroup \em et al.\egroup
  }{2015}]{BaCS15}
Vince B{\'{a}}r{\'{a}}ny, Balder ten Cate, and Luc Segoufin.
\newblock Guarded negation.
\newblock {\em J. {ACM}}, 62(3):22:1--22:26, 2015.

\bibitem[\protect\citeauthoryear{Barcel{\'{o}} \bgroup \em et al.\egroup
  }{2014}]{BaRV14}
Pablo Barcel{\'{o}}, Miguel Romero, and Moshe~Y. Vardi.
\newblock Does query evaluation tractability help query containment?
\newblock In {\em PODS}, pages 188--199, 2014.

\bibitem[\protect\citeauthoryear{Benedikt \bgroup \em et al.\egroup
  }{2015}]{BeCCB15}
Michael Benedikt, Balder ten Cate, Thomas Colcombet, and Michael {Vanden Boom}.
\newblock The complexity of boundedness for guarded logics.
\newblock In {\em LICS}, pages 293--304, 2015.

\bibitem[\protect\citeauthoryear{Bienvenu \bgroup \em et al.\egroup
  }{2013}]{BiLW13}
Meghyn Bienvenu, Carsten Lutz, and Frank Wolter.
\newblock First-order rewritability of atomic queries in horn description
  logics.
\newblock In {\em IJCAI}, 2013.

\bibitem[\protect\citeauthoryear{Bienvenu \bgroup \em et al.\egroup
  }{2014}]{BCLW14}
Meghyn Bienvenu, Balder ten Cate, Carsten Lutz, and Frank Wolter.
\newblock Ontology-based data access: {A} study through disjunctive datalog,
  {CSP}, and {MMSNP}.
\newblock {\em {ACM} Trans. Database Syst.}, 39(4):33:1--33:44, 2014.

\bibitem[\protect\citeauthoryear{Bienvenu \bgroup \em et al.\egroup
  }{2016}]{BiHLW16}
Meghyn Bienvenu, Peter Hansen, Carsten Lutz, and Frank Wolter.
\newblock First order-rewritability and containment of conjunctive queries in
  horn description logics.
\newblock In {\em IJCAI}, pages 965--971, 2016.

\bibitem[\protect\citeauthoryear{Blumensath \bgroup \em et al.\egroup
  }{2014}]{BOW14}
Achim Blumensath, Martin Otto, and Mark Weyer.
\newblock Decidability results for the boundedness problem.
\newblock {\em Logical Methods in Computer Science}, 10(3), 2014.

\bibitem[\protect\citeauthoryear{Cal{\`{\i}} \bgroup \em et al.\egroup
  }{2012a}]{CaGL12}
Andrea Cal{\`{\i}}, Georg Gottlob, and Thomas Lukasiewicz.
\newblock A general datalog-based framework for tractable query answering over
  ontologies.
\newblock {\em J. Web Sem.}, 14:57--83, 2012.

\bibitem[\protect\citeauthoryear{Cal{\`{\i}} \bgroup \em et al.\egroup
  }{2012b}]{CaGP12}
Andrea Cal{\`{\i}}, Georg Gottlob, and Andreas Pieris.
\newblock Towards more expressive ontology languages: The query answering
  problem.
\newblock {\em Artif. Intell.}, 193:87--128, 2012.

\bibitem[\protect\citeauthoryear{Cal\`{\i} \bgroup \em et al.\egroup
  }{2013}]{CaGK13}
Andrea Cal\`{\i}, Georg Gottlob, and Michael Kifer.
\newblock Taming the infinite chase: Query answering under expressive
  relational constraints.
\newblock {\em J.~Artif.~Intell.~Res.}, 48:115--174, 2013.

\bibitem[\protect\citeauthoryear{Colcombet and Fijalkow}{2016}]{CoF16}
Thomas Colcombet and Nathana{\"{e}}l Fijalkow.
\newblock The bridge between regular cost functions and omega-regular
  languages.
\newblock In {\em ICALP}, pages 126:1--126:13, 2016.

\bibitem[\protect\citeauthoryear{Cosmadakis \bgroup \em et al.\egroup
  }{1988}]{CGKV88}
Stavros~S. Cosmadakis, Haim Gaifman, Paris~C. Kanellakis, and Moshe~Y. Vardi.
\newblock Decidable optimization problems for database logic programs
  (preliminary report).
\newblock In {\em STOC}, pages 477--490, 1988.

\bibitem[\protect\citeauthoryear{Fagin \bgroup \em et al.\egroup
  }{2005}]{FKMP05}
Ronald Fagin, Phokion~G. Kolaitis, Ren{\'{e}}e~J. Miller, and Lucian Popa.
\newblock Data exchange: semantics and query answering.
\newblock {\em Theor. Comput. Sci.}, 336(1):89--124, 2005.

\bibitem[\protect\citeauthoryear{Gaifman \bgroup \em et al.\egroup
  }{1993}]{GMSV93}
Haim Gaifman, Harry~G. Mairson, Yehoshua Sagiv, and Moshe~Y. Vardi.
\newblock Undecidable optimization problems for database logic programs.
\newblock {\em J. {ACM}}, 40(3):683--713, 1993.

\bibitem[\protect\citeauthoryear{Gottlob \bgroup \em et al.\egroup
  }{2014}]{GoOP14}
Georg Gottlob, Giorgio Orsi, and Andreas Pieris.
\newblock Query rewriting and optimization for ontological databases.
\newblock {\em {ACM} Trans. Database Syst.}, 39(3):25:1--25:46, 2014.

\bibitem[\protect\citeauthoryear{Lutz and Sabellek}{2017}]{LuSa17}
Carsten Lutz and Leif Sabellek.
\newblock Ontology-mediated querying with the description logic {EL:}
  trichotomy and linear datalog rewritability.
\newblock In {\em IJCAI}, pages 1181--1187, 2017.

\bibitem[\protect\citeauthoryear{Poggi \bgroup \em et al.\egroup
  }{2008}]{PLCG+08}
Antonella Poggi, Domenico Lembo, Diego Calvanese, Giuseppe {De Giacomo},
  Maurizio Lenzerini, and Riccardo Rosati.
\newblock Linking data to ontologies.
\newblock {\em J. Data Semantics}, 10:133--173, 2008.

\bibitem[\protect\citeauthoryear{Rossman}{2008}]{Ro08}
Benjamin Rossman.
\newblock Homomorphism preservation theorems.
\newblock {\em J. {ACM}}, 55(3):15:1--15:53, 2008.

\bibitem[\protect\citeauthoryear{Vardi}{1992}]{Va92}
Moshe~Y. Vardi.
\newblock Automata theory for database theoreticans.
\newblock In {\em Theoretical Studies in Computer Science}, pages 153--180,
  1992.

\bibitem[\protect\citeauthoryear{Vardi}{1998}]{Va98}
Moshe~Y. Vardi.
\newblock Reasoning about the past with two-way automata.
\newblock In {\em ICALP}, pages 628--641, 1998.

\end{thebibliography}

\end{document}